\documentclass[11pt,twoside]{article}
\usepackage{a4}

\usepackage{amssymb,amsmath,amsthm,latexsym}
\usepackage{amsfonts}
\usepackage{amsfonts}
\usepackage{graphicx}

\usepackage{amsmath, amsfonts}
\usepackage{amssymb, graphicx}
\usepackage{amscd}
\usepackage{textcomp}
\usepackage{palatino}
\usepackage{xcolor}
\usepackage{array}
\usepackage{multirow}
\usepackage{enumitem}
\usepackage{comment}
\usepackage[colorlinks=true,linkcolor=red,citecolor=red]{hyperref}
\allowdisplaybreaks

\newtheorem{theorem}{Theorem}[section]

\newtheorem{corollary}[theorem] {Corollary}
\newtheorem{definition}[theorem]{Definition}
\newtheorem{example}[theorem]{Example}
\newtheorem{lemma} [theorem]{Lemma}

\newtheorem{proposition}[theorem]{Proposition}
\newtheorem{remark}[theorem]{Remark}

\newcommand{\Z}{{\mathbb Z}}

\newcommand{\F}{{\mathbb F}}

\vspace{5cm}

\begin{document}
	\label{'ubf'}  
	\setcounter{page}{1} 
	\markboth
	{\hspace*{-9mm} \centerline{\footnotesize
			Quadratic Residue Codes over $\Z_{121}$}}
	{\centerline{\footnotesize 
			 Tapas Chatterjee and Priya Jain } \hspace*{-9mm}}
	
	\vspace*{-2cm}

	\begin{center}
	{{\textbf{Quadratic Residue Codes over $\Z_{121}$ 
		}}\\
		\vspace{.2cm}
		\medskip
		\sc Tapas Chatterjee 
	\\
	{\footnotesize  Department of Mathematics,}\\
	{\footnotesize Indian Institute of Technology Ropar, Punjab, India.}\\
	{\footnotesize e-mail: {\it tapasc@iitrpr.ac.in}}
	
	\medskip
	{\sc Priya Jain\footnote{Research of the second author is supported by University Grants Commission (UGC), India , Ref No.: 231610153762/(CSIR-UGC NET DEC-2022/JUNE-2023(Merged Cycle)).}}\\
	{\footnotesize Department of Mathematics, }\\
	{\footnotesize Indian Institute of Technology Ropar, Punjab, India.}\\
	{\footnotesize e-mail: {\it priya.23maz0012@iitrpr.ac.in}}
	\medskip}
\end{center}

\thispagestyle{empty} 
\vspace{-.4cm}

\hrulefill
\noindent

\begin{abstract}  
	{\footnotesize }
	In this paper, we construct a special family of cyclic codes, known as quadratic residue codes of prime length \( p \equiv \pm 1 \pmod{44} ,\) \( p \equiv \pm 5 \pmod{44} ,\) \( p \equiv \pm 7 \pmod{44} ,\) \( p \equiv \pm 9 \pmod{44} \) and \( p \equiv \pm 19 \pmod{44} \) over $\Z_{121}$ by defining them using their generating idempotents. Furthermore, the properties of these codes and extended quadratic residue codes over $\Z_{121}$ are discussed, followed by their Gray images. Also, we show that the extended quadratic residue code over $\mathbb{Z}_{121}$ possesses a large permutation automorphism group generated by shifts, multipliers, and inversion, making permutation decoding feasible. As examples, we construct new codes with parameters $[55,5,33]$ and $[77,7,44].$  
\end{abstract}
\hrulefill

\noindent 

{\small \textbf{Keywords and phrases}: Automorphism group of a code, Cyclic codes, Extended quadratic residue codes, Generating Idempotents, Quadratic residue codes}.

\noindent
{\bf{MSC 2020:}} Primary: 11T71, 94B15, 94B60; Secondary: 94B05, 51E22, 94B40  
\vspace{-.37cm}

\tableofcontents

\section{\bf Introduction}

\hspace{5mm} Let $\Z_{121}$ denote the ring of integers modulo $121.$ A subset of n-tuples over $\Z_{121}$ is called a linear code over $\Z_{121}$ or a $\Z_{121}$-code, if it is a $\Z_{121}$-submodule of $\Z_{121}^n.$ A cyclic code $C$ of length $n$ over a commutative ring $R$ is an ideal of $\displaystyle \frac{R[x]}{<x^n-1>}$. A unique class of prime-length cyclic codes known as quadratic residue (QR) codes was introduced to construct self-dual codes by including an overall parity check. Andrew Gleason was the first to define and study QR codes. Examples are the binary Hamming $[7,4,3]$ code, the binary Golay $[23,12,7]$ code, the ternary Golay $[11,6,5]$ code, and so on.

The class of quadratic residue codes is crucial to the theory of error-correcting codes because many of them have minimum weights that are high relative to their lengths. While QR codes over finite fields have been extensively studied, more recent study has concentrated on these codes over finite rings. Chiu et al. and Taeri examined the structure of QR codes over $\mathbb{Z}_8$ \cite{CY} and $\mathbb{Z}_9$ \cite{BT}, respectively, while Pless and Qian examined QR codes over $\mathbb{Z}_4$ in \cite{PQ}. The structure of QR codes over $\mathbb{F}_p+u\mathbb{F}_p$ \cite{KY} and $\mathbb{F}_2+u\mathbb{F}_2+u^2\mathbb{F}_2$ \cite{AB} was introduced by Kaya et al. The structure of these codes over a non-local ring $\mathbb{F}_p+u\mathbb{F}_p+u^2\mathbb{F}_p$, where $p$ is an odd prime and $u^3=u$ was examined by Liu et al in \cite{LS}. M. Raka et al. expanded their findings across the ring $\mathbb{F}_p+u\mathbb{F}_p+u^2\mathbb{F}_p+u^3\mathbb{F}_p$ where $p\equiv 1 \pmod 3$ and $u^4=u$ in \cite{RG}. The automorphism groups of the extended QR codes over $\mathbb{Z}_{16}$ and $\mathbb{Z}_{32}$ were then explored by Hsu et al. in \cite{HL}. This paper is devoted to the study of QR codes over $\mathbb{Z}_{121}.$ We demonstrate the large automorphism group of these quadratic residue codes over $\mathbb{Z}_{121},$ which will be useful in decoding these codes, utilizing permutation decoding techniques. Likewise, the authors also constructed the extended binary Hamming code $[8,4,4]$ using nilpotent elements in the group ring $\mathbb{F}_2[D_8]$ and as a block circulant polynomial code in their previous works \cite{CJ},\cite{CP}.

\section{\bf{Overview and Background\label{sec: overview}}}
In this section, we recall some basic notations and definitions related to quadratic residue, quadratic residue codes, idempotent generators, etc., along with several useful results essential for the subsequent sections.

Let $\mathbb{F}_q$ denote the finite field with $q=p^m$ elements where $p$ is a prime number and $m$ is a positive integer. 

\begin{definition}
    A $linear$ code $C$ of length $n$ over $\F_q$ is a subspace of $\F_q^n,$ denoted by $[n,k,d],$ where $n$ is the length, $k$ is the dimension and $d$ is the Hamming distance of the code.
\end{definition}

\begin{definition}
    If $C$ is an $[n,k]$ code over the field $\F_q,$ its dual code $C^\perp$ is defined as the set of vectors that are orthogonal to all the codewords of $C.$ Moreover, $C^\perp$ is an $[n,n/2]$ code.

If $C$ = $C^\perp, C$ is called self-dual code.
\end{definition}

In this article, we examine the quadratic residue codes over $\Z_{121}.$ To construct, we use idempotent generators \cite{MS}.

\begin{definition}
    Let $l$ be an odd prime and $m$ be another prime which is quadratic residue modulo $l.$ Let $\mathcal{Q}$ denote the set of quadratic residues modulo $l$ and $\mathcal{N}$ be the set of non-residues. Let $\displaystyle r(x)=\prod_{i\in \mathcal{Q}}(x-\alpha^i)$ and $\displaystyle n(x)=\prod_{i\in \mathcal{N}}(x-\alpha^i),$ where $\alpha$ is a primitive $l^{th}$ root of unity in some extension field of $\F_m.$ Let $Q, N$ be the $QR$ codes generated by $r(x)$ and $n(x)$ over $\F_m$ and $\bar{Q}, \bar{N}$ be the $QR$ codes generated by $(x-1)r(x)$ and $(x-1)n(x)$ respectively. Suppose $l \equiv \pm 1\pmod{4}$ and $m$ is a quadratic residue modulo $l.$ Let $\displaystyle e_1=\sum_{i\in \mathcal{Q}}x^i,$ $\displaystyle e_2=\sum_{i\in \mathcal{N}}x^i$ and $h=1+e_1+e_2$ is the polynomial corresponding to the all one vector of length $l.$ Then idempotent generators\label{def: idempotent} of quadratic residue codes $Q, N, \bar{Q}$ and $\bar{N}$ over $\F_m$ are:
\begin{align*}
E_q(x)=\frac{1}{2}(1+\frac{1}{l})+\frac{1}{2}(\frac{1}{l}-\frac{1}{\theta})e_1(x)+\frac{1}{2}(\frac{1}{l}+\frac{1}{\theta})e_2(x),\\
E_n(x)=\frac{1}{2}(1+\frac{1}{l})+\frac{1}{2}(\frac{1}{l}-\frac{1}{\theta})e_2(x)+\frac{1}{2}(\frac{1}{l}+\frac{1}{\theta})e_1(x),\\
\bar{E}_q(x)=\frac{1}{2}(1-\frac{1}{l})-\frac{1}{2}(\frac{1}{l}+\frac{1}{\theta})e_1(x)-\frac{1}{2}(\frac{1}{l}-\frac{1}{\theta})e_2(x),\\
\bar{E}_n(x)=\frac{1}{2}(1-\frac{1}{l})-\frac{1}{2}(\frac{1}{l}+\frac{1}{\theta})e_2(x)-\frac{1}{2}(\frac{1}{l}-\frac{1}{\theta})e_1(x),\\
\end{align*}
respectively, where $\displaystyle \theta=\sum_{i=1}^{l-1}\chi(i)\alpha^i$ denotes the Gaussian sum, $\chi(i)$ denotes the Legendre symbol and
\begin{equation}\label{def:gaussian}
\theta^2= \begin{cases}
l & \text{if } l \equiv 1\,(\!\bmod \ 4)
 \ ,\\ 
-l & \text{if } l \equiv -1\,(\!\bmod \ 4)
\end{cases} 
\qquad 
and 
\qquad
\chi(i)= 
\begin{cases}
 0 & \text{if } l|i,  \\
 1 & \text{if } i \in Q, \\
-1 & \text{if } i \in N.
\end{cases}
\end{equation}
\end{definition}

Quadratic residue codes are well-known cyclic codes, first introduced by Gleason in his 1958 paper. This paper laid the foundation for the development of these codes, which are a type of cyclic code used in error-correcting coding theory. The following theorem gives the evaluation of the Legendre symbol $\displaystyle \left(\frac{11}{p}\right),$ where $p$ is an odd prime.

\begin{theorem}\label{thm:4.1}
    The Legendre symbol $\displaystyle \left(\frac{11}{p}\right)=1$ if and only if \\ \[p\equiv 1,5,7,9,19,25,35,37,39 \ \text{or} \ 43 \,(\!\bmod{44}).\]
\end{theorem}
\begin{proof}
Using Law of Quadratic Reciprocity, we have $\displaystyle \left(\frac{11}{p}\right)=(-1)^{\frac{11-1}{2}\frac{p-1}{2}}\left(\frac{p}{11}\right).$\\
\[i.e.,\displaystyle   \left(\frac{11}{p}\right)=(-1)^{\frac{p-1}{2}}\left(\frac{p}{11}\right).\]\\
For $\displaystyle \left(\frac{11}{p}\right)=1,$ both factors in RHS ($i.e. \ (-1)^{\frac{p-1}{2}}$ and $\displaystyle \left(\frac{p}{11}\right)$) should be of same parity ($i.e.,1$ or $-1$).\\
\begin{enumerate}
    \item[\textbf{Case 1:}] If both factors are $1,$ then 
    \[
    \begin{aligned}
    (-1)^{\frac{p-1}{2}}=1 \ \ \text{implies} \ \ p \equiv 1 \,(\!\bmod \ 4) \  \text{and} \\
    \displaystyle \left(\frac{p}{11}\right)=1 \ \ \text{implies} \ \ p \equiv 1,3,4,5 \ \text{or} \ 9 \,(\!\bmod \ 11).
    \end{aligned}
    \]
    
    Combining each pair and using Chinese Remainder Theorem:-
    \begin{itemize}
        \item $p \equiv 1 \,(\!\bmod \ 4)$ and $p \equiv 1 \,(\!\bmod \ 11) \ \ \text{together gives} \ \ p \equiv 1 \,(\!\bmod \ 44)$
        \item $p \equiv 1 \,(\!\bmod \ 4)$ and $p \equiv 3 \,(\!\bmod \ 11) \ \ \text{together gives} \ \ p \equiv 25 \,(\!\bmod \ 44)$
        \item $p \equiv 1 \,(\!\bmod \ 4))$ and $p \equiv 4 \,(\!\bmod \ 11) \ \ \text{together gives} \ \ p \equiv 37 \,(\!\bmod \ 44)$
        \item $p \equiv 1 \,(\!\bmod \ 4)$ and $p \equiv 5 \,(\!\bmod \ 11) \ \ \text{together gives} \ \ p \equiv 5 \,(\!\bmod \ 44)$
        \item $p \equiv 1 \,(\!\bmod \ 4)$ and $p \equiv 9 \,(\!\bmod \ 11) \ \ \text{together gives} \ \ p \equiv 9 \,(\!\bmod \ 44).$
    \end{itemize}
    That is, $p \equiv 1,5,9,25 \ \text{or} \ 37 \,(\!\bmod \ 44).$

        \item[\textbf{Case 2:}] If both factors are $-1,$ then 
    \[
    \begin{aligned}
    (-1)^{\frac{p-1}{2}}=-1 \ \ \text{implies} \ \ p \equiv 3 \,(\!\bmod \ 4) \ \text{and} \\
    \left(\frac{p}{11}\right)=-1 \ \ \text{implies} \ \ p \equiv 2,6,7,8 \ \text{or} \ 10 \,(\!\bmod \ 11).
    \end{aligned}
    \]
    Combining each pair and using Chinese Remainder Theorem:-
    \begin{itemize}
        \item $p \equiv 3 \,(\!\bmod \ 4)$ and $p \equiv 2 \,(\!\bmod \ 11) \ \ \text{together gives} \ \ p \equiv 35 \,(\!\bmod \ 44)$
        \item $p \equiv 3 \,(\!\bmod \ 4)$ and $p \equiv 6 \,(\!\bmod \ 11) \ \ \text{together gives} \ \ p \equiv 39 \,(\!\bmod \ 44)$
        \item $p \equiv 3 \,(\!\bmod \ 4)$ and $p \equiv 7 \,(\!\bmod \ 11) \ \ \text{together gives} \ \ p \equiv 7 \,(\!\bmod \ 44)$
        \item $p \equiv 3 \,(\!\bmod \ 4)$ and $p \equiv 8 \,(\!\bmod \ 11) \ \ \text{together gives} \ \ p \equiv 19 \,(\!\bmod \ 44)$
        \item $p \equiv 3 \,(\!\bmod \ 4)$ and $p \equiv 10 \,(\!\bmod \ 11) \ \ \text{together gives} \ \ p \equiv 43 \,(\!\bmod \ 44).$
    \end{itemize}
    That is, $p \equiv 7,19,35,39 \ \text{or} \ 43 \,(\!\bmod \ 44).$
\end{enumerate}

   Hence, $\displaystyle \left(\frac{11}{p}\right)=1$ if and only if $p\equiv \pm1,\pm5,\pm7, \pm9 \ \text{or} \ \pm19 \,(\!\bmod \ 44).$
\end{proof}

According to Theorem \ref{thm:4.1}, in order to examine quadratic residue codes over $\Z_{11},$ we must suppose that $p\equiv\pm1\pmod{44},$ $p\equiv\pm5\pmod{44},$ $p\equiv\pm7\pmod{44},$ $p\equiv\pm9\pmod{44},$ and $p\equiv\pm19\pmod{44}.$ The following are the lemmas that give the idempotent elements corresponding to each case.

\begin{lemma}
    Let $p$ be a prime of the form $p\equiv\pm1\pmod{44},$ then the set of idempotents of $\displaystyle \frac{\Z_{11}[x]}{<x^p-1>}$ is $\{1+e_1, 1+e_2, -e_1, -e_2\}.$
\end{lemma}
\begin{proof} We examine each of the two cases separately.

\textbf{Case 1:} If $p=44k+1,$ then $p\equiv1 \pmod{4}.$
    Therefore, by using \ref{def:gaussian}, we have $\theta^2=p\equiv1 \pmod{11}$ and hence $\theta= \pm 1.$
    \begin{itemize}
        \item For $\theta=1,$ using \ref{def: idempotent}, we have the following idempotent generators:
    \[
        E_q(x)=1+e_2, \ \  E_n(x)=1+e_1, \ \  \bar{E}_q(x)=-e_1, \ \ \bar{E}_n(x)=-e_2.
    \] 

       \item For $\theta=-1,$ we have the following idempotent generators:
    \[
        E_q(x)=1+e_1, \ \ E_n(x)=1+e_2, \ \ \bar{E}_q(x)=-e_2, \ \  \bar{E}_n(x)=-e_1.
    \]
    \end{itemize} 

    \textbf{Case 2:} If $p=44k-1,$ then $p\equiv-1 \pmod{4}.$
    Therefore, by using \ref{def:gaussian}, we have $\theta^2=-p\equiv1 \pmod{11}$ and hence $\theta= \pm 1.$
    \begin{itemize}
    \item For $\theta=1,$ we have:
    \[
        E_q(x)=-e_1, \ \  E_n(x)=-e_2, \ \ \bar{E}_q(x)=1+e_2, \ \ \bar{E}_n(x)=1+e_1.
    \]
    
    \item For $\theta=-1,$ we have:
    \[
        E_q(x)=-e_2, \ \  E_n(x)=-e_1, \ \  \bar{E}_q(x)=1+e_1, \ \ \bar{E}_n(x)=1+e_2.
    \]
    \end{itemize}
\end{proof}

\begin{lemma}
    Let $p$ be a prime of the form $p\equiv\pm5\pmod{44}, \ p\equiv\pm7\pmod{44}, p\equiv\pm9\pmod{44}$ and $p\equiv\pm19\pmod{44},$ then the set of idempotents of $\displaystyle \frac{\Z_{11}[x]}{<x^p-1>}$ are $\{5+3e_1+6e_2, 5+6e_1+3e_2, 7+5e_1+8e_2, 7+8e_1+5e_2\}, \ \{2+4e_1+10e_2, 2+10e_1+4e_2, 10+e_1+7e_2, 10+7e_1+e_2\}, \{3+6e_1+10e_2, 3+10e_1+6e_2, 9+e_1+5e_2, 9+5e_1+e_2\}$ and $\{4+10e_1+8e_2, 4+8e_1+10e_2, 8+3e_1+e_2, 8+e_1+3e_2\}$ respectively.
\end{lemma}
\begin{proof} It suffices to verify the cases where $p=44k+5,\ p=44k+7,\ p=44k+9,$ and $p=44k+19,$ the rest cases follow in a similar way.
\begin{enumerate}
\item When $p=44k+5,$ then $p\equiv1 \pmod{4}.$ Therefore, by using \ref{def:gaussian}, we have $\theta^2=p\equiv5 \pmod{11}$ and hence $\theta= \pm 4.$
   \begin{itemize} 
    \item For $\theta=4,$ we have:
    \[
        E_q(x)=5+3e_1+6e_2, \  E_n(x)=5+6e_1+3e_2,  \  \bar{E}_q(x)=7+5e_1+8e_2,  \     \bar{E}_n(x)=7+8e_1+5e_2.
    \]
    
    \item For $\theta=-4,$ we have:
    \[
        E_q(x)=5+6e_1+3e_2, \  E_n(x)=5+3e_1+6e_2, \ \bar{E}_q(x)=7+8e_1+5e_2, \       \bar{E}_n(x)=7+5e_1+8e_2.
    \]
    \end{itemize}

\item When $p=44k+7,$ then $p\equiv-1 \pmod{4}.$ Therefore, by using \ref{def:gaussian}, we have $\theta^2=-p\equiv4 \pmod{11}$ and hence $\theta= \pm 2.$
    
    \begin{itemize}
        \item For $\theta=2,$ we have:
    \[
        E_q(x)=2+4e_1+10e_2,\ 
        E_n(x)=2+10e_1+4e_2,\
        \bar{E}_q(x)=10+e_1+7e_2,\
        \bar{E}_n(x)=10+7e_1+e_2.
    \]
    
    \item For $\theta=-2,$ we have:
    \[
        E_q(x)=2+10e_1+4e_2,\
        E_n(x)=2+4e_1+10e_2,\
        \bar{E}_q(x)=10+7e_1+e_2,\
        \bar{E}_n(x)=10+e_1+7e_2.
    \]
    \end{itemize}    

\item When $p=44k+9,$ then $p\equiv1 \pmod{4}.$ Therefore, by using \ref{def:gaussian}, we have $\theta^2=p\equiv9 \pmod{11}$ and hence $\theta= \pm 3.$
    \begin{itemize}
    \item For $\theta=3,$ we have:
    \[
        E_q(x)=3+6e_1+10e_2,\
        E_n(x)=3+10e_1+6e_2,\
        \bar{E}_q(x)=9+e_1+5e_2,\
        \bar{E}_n(x)=9+5e_1+e_2.
    \]
    
    \item For $\theta=-3,$ we have:
    \[
        E_q(x)=3+10e_1+6e_2,\
        E_n(x)=3+6e_1+10e_2,\
        \bar{E}_q(x)=9+5e_1+e_2,\
        \bar{E}_n(x)=9+e_1+5e_2.
    \]
    \end{itemize}

\item When $p=44k+19,$ then $p\equiv-1 \pmod{4}.$ Therefore, by using \ref{def:gaussian}, we have $\theta^2=-p\equiv3 \pmod{11}$ and hence $\theta= \pm 5.$   
     \begin{itemize}
    \item For $\theta=5,$ we have:
    \[
        E_q(x)=4+10e_1+8e_2,\
        E_n(x)=4+8e_1+10e_2,\
        \bar{E}_q(x)=8+3e_1+e_2,\
        \bar{E}_n(x)=8+e_1+3e_2.
    \]
    
    \item For $\theta=-5,$ we have:
    \[
        E_q(x)=4+8e_1+10e_2,\
        E_n(x)=4+10e_1+8e_2,\
        \bar{E}_q(x)=8+e_1+3e_2,\
        \bar{E}_n(x)=8+3e_1+e_2.
    \]
\end{itemize}
\end{enumerate}
\end{proof}

A $\Z_{11}$- cyclic code is called a $\Z_{11}$- quadratic residue code if it is generated by one of the above idempotents. Using the following theorem from \cite{PQ}, we can also find the idempotent elements of the dual code. 

\begin{theorem}\label{thm:4.7}
    \begin{enumerate}[label=(\roman*)]
        \item Let $C$ be a cyclic code of length $p$ over $\Z_{121}$ generated by the idempotent $E_q(x)$ in quotient ring $\displaystyle \frac{\Z_{11}[x]}{<x^p-1>},$ then $C^\perp$ is generated by the idempotent $1-E_q(x^{-1}).$
        \item Let $C_1$ and $C_2$ be cyclic codes of length $p$ over $\Z_{121}$ generated by the idempotents $E_1$ and $E_2$ in quotient ring $\displaystyle \frac{\Z_{11}[x]}{<x^p-1>},$ respectively, then $C_1\cap C_2$ and $C_1+C_2$ are generated by the idempotents $E_1E_2$ and $E_1+E_2-E_1E_2,$ respectively.
    \end{enumerate}
\end{theorem}


The following theorem from \cite{JK} gives the relationship between $e_1$ and $e_2$ depending on $p.$

\begin{theorem} Let $\displaystyle R=\frac{\Z_q[x]}{<x^p-1>}$ and $\displaystyle e_1=\sum_{i\in \mathcal{Q}}x^i, e_2=\sum_{i\in \mathcal{N}}x^i$ be two elements in $R,$ where $\mathcal{Q}$ and $\mathcal{N}$ denote the set of quadratic residues and the set of quadratic non-residues modulo $p.$
\begin{itemize}
    \item If $p\equiv 1 \pmod{4},$ then we have
    \[
    \begin{aligned}
        e_1^2=\left(\frac{p-5}{4}\right)e_1+\left(\frac{p-1}{4}\right)e_2+\left(\frac{p-1}{2}\right)\\
        e_2^2=\left(\frac{p-1}{4}\right)e_1+\left(\frac{p-5}{4}\right)e_2+\left(\frac{p-1}{2}\right)\\
        e_1e_2=\left(\frac{p-1}{4}\right)e_1+\left(\frac{p-1}{4}\right)e_2 \ .
    \end{aligned}
    \]

    \item If $p\equiv 3 \pmod{4},$ then we have
    \[
    \begin{aligned}
        e_1^2=\left(\frac{p-3}{4}\right)e_1+\left(\frac{p+1}{4}\right)e_2\\
        e_2^2=\left(\frac{p+1}{4}\right)e_1+\left(\frac{p-3}{4}\right)e_2\\
        e_1e_2=\left(\frac{p-1}{2}\right)+\left(\frac{p-3}{4}\right)e_1+\left(\frac{p-3}{4}\right)e_2 \ .
    \end{aligned}
    \]
\end{itemize}    
\end{theorem}
\begin{remark}
    In $\displaystyle \frac{\Z_q[x]}{<x^p-1>}, h^2=(1+e_1+e_2)h=h+\left(\frac{p-1}{2}\right)h+\left(\frac{p-1}{2}\right)h=ph.$
\end{remark}

\section{\bf{Quadratic residue codes over \texorpdfstring{\ensuremath{\Z_{121}}}{Z121}}\label{sec: classification}}
\hspace{5mm} This section examines QR codes over $\Z_{121}$ and presents various findings on them. They are described in terms of their idempotent generators. Throughout the paper, we assume that $p$ is an odd prime and $h=1+e_1+e_2$ is the polynomial corresponding to the all-one vector. The following theorem in \cite{PK} gives the idempotent generators of the $q$-adic quadratic residue codes of length $p.$

\begin{theorem}
    Let $q$ be an odd prime and $p=4k\pm1$ be a prime such that $q$ is a quadratic residue modulo $p.$ Let $\theta\equiv\pm1\pmod{q},$ where double signs are in the same order as in $p=4k\pm1.$ The idempotent generators of the $q$-adic QR codes
    \[
    <Q_\infty(x)>, <N_\infty(x)>, <(x-1)Q_\infty(x)>, <(x-1)N_\infty(x)>
    \]
    of length $p$ are given as follows, respectively:
    \begin{align*}
        E_q(x)=a+be_1+ce_2\\
        E_n(x)=a+ce_1+be_2\\
        \bar{E}_q(x)=a'-ce_1-be_2\\
        \bar{E}_n(x)=a'-be_1-ce_2        
    \end{align*}
    where \[ a=\frac{p+1}{2p}, a'=\frac{p-1}{2p}, b=\frac{1\mp\theta}{2p}, c=\frac{1\pm\theta}{2p}.\]
\end{theorem}

For a prime $p\equiv 1, 5, 7, 9, 19, 25, 35, 37, 39 \ \text{or} \ 43 \pmod{44},$ a few forthcoming theorems provide the categorization of idempotents in $\displaystyle \frac{\Z_{121}[x]}{<x^p-1>}.$

\begin{theorem}
    Let $p$ be a prime of the form $44k+1,$ then the collection of idempotents in $\displaystyle \frac{\Z_{121}[x]}{<x^p-1>}$ corresponding to each case depending on $k$ is given as follows:
    \begin{enumerate}[label=(\roman*)]
    \item If $k=11t,$ then $\{1+e_2, 1+e_1, -e_1, -e_2, h, 1-h\}$ are idempotents.
    \item If $k=11t+1,$ then $\{100+110e_1+89e_2, 100+89e_1+110e_2, 22+32e_1+11e_2, 22+11e_1+32e_2,78h,1-78h\}$ are idempotents.
    \item If $k=11t+2,$ then $\{78+99e_1+56e_2, 78+56e_1+99e_2, 44+65e_1+22e_2, 44+22e_1+65e_2, 34h, 1-34h\}$ are idempotents.
    \item If $k=11t+3,$ then $\{56+23e_1+88e_2, 56+88e_1+23e_2,66+33e_1+98e_2,66+98e_1+33e_2,111h,1-111h\}$ are idempotents.
    \item If $k=11t+4,$ then $\{34+111e_1+77e_2,34+77e_1+111e_2,88+44e_1+10e_2,88+10e_1+44e_2,67h,1-67h\}$ are idempotents.
    \item If $k=11t+5,$ then $\{12+78e_1+66e_2,12+66e_1+78e_2,110+55e_1+43e_2,110+43e_1+55e_2,23h,1-23h\}$ are idempotents.
    \item If $k=11t+6,$ then $\{111+55e_1+45e_2,111+45e_1+55e_2,11+76e_1+66e_2,11+66e_1+76e_2,100h,1-100h\}$ are idempotents.
    \item If $k=11t+7,$ then $\{89+44e_1+12e_2,89+12e_1+44e_2,33+109e_1+77e_2,33+77e_1+109e_2,56h,1-56h\}$ are idempotents.
    \item If $k=11t+8,$ then $\{67+33e_1+100e_2,67+100e_1+33e_2,55+21e_1+88e_2,55+88e_1+21e_2,12h,1-12h\}$ are idempotents.
    \item If $k=11t+9,$ then $\{45+67e_1+22e_2,45+22e_1+67e_2,77+99e_1+54e_2,77+54e_1+99e_2,89h,1-89h\}$ are idempotents.
    \item If $k=11t+10,$ then $\{23+34e_1+11e_2,23+11e_1+34e_2,99+110e_1+87e_2,99+87e_1+110e_2,45h,1-45h\}$ are idempotents.
    \end{enumerate}
\end{theorem}
\begin{proof}
    Let $p=44k+1.$ We compute everything modulo $121.$
    \begin{enumerate}
    \item[\textbf{Case 1:}]  If $k=11t,$ then $p=1.$ The inverse of $2p=2$ is $61.$\\
    Solving $\theta^2=p=1,$ we get $\theta = \pm 1.$\\
    Thus, $a=61(p+1)=1, \ a'=61(p-1)=0$ and $b,c=61(1\pm \theta)=1,0.$\\
    Hence, idempotent generators of QR codes are $1+e_2,1+e_1,-e_1,-e_2,h,1-h.$

    \item[\textbf{Case 2:}]  If $k=11t+1,$ then $p=45.$ The inverse of $2p=90$ is $39.$\\
    Solving $\theta^2=p=45=529,$ we get $\theta = \pm 23.$\\
    Thus, $a=39(p+1)=100, \ a'=39(p-1)=22$ and  $b,c=39(1\pm \theta)=110,89.$\\
    Hence, idempotent generators of QR codes are $100+110e_1+89e_2,100+89e_1+110e_2,22+32e_1+11e_2,22+11e_1+32e_2,78h,1-78h.$\\
    

    \item[\textbf{Case 3:}]  If $k=11t+2,$ then $p=89.$ The inverse of $2p=57$ is $17.$\\
    Solving $\theta^2=p=89=2025,$ we get $\theta = \pm 45.$\\
    Thus, $a=17(p+1)=78, \ a'=17(p-1)=44$ and $b,c=17(1\pm \theta)=99,56.$\\
    Hence, idempotent generators of QR codes are $78+99e_1+56e_2,78+56e_1+99e_2,44+65e_1+22e_2,44+22e_1+65e_2,34h,1-34h.$

    \item[\textbf{Case 4:}]  If $k=11t+3,$ then $p=12.$ The inverse of $2p=24$ is $116.$\\
    Solving $\theta^2=p=12=2916,$ we get $\theta = \pm 54.$\\
    Thus, $a=116(p+1)=56, \ a'=116(p-1)=66$ and $b,c=116(1\pm \theta)=23,88.$\\
    Hence, idempotent generators of QR codes are $56+23e_1+88e_2,56+88e_1+23e_2,66+33e_1+98e_2,66+98e_1+33e_2,111h,1-111h.$

    \item[\textbf{Case 5:}]  If $k=11t+4,$ then $p=56.$ The inverse of $2p=112$ is $94.$\\
    Solving $\theta^2=p=56=1024,$ we get $\theta = \pm 32.$\\
    Thus, $a=94(p+1)=34, \ a'=94(p-1)=88$ and $b,c=94(1\pm \theta)=111,77.$\\
    Hence, idempotent generators of QR codes are $34+111e_1+77e_2,34+77e_1+111e_2,88+44e_1+10e_2,88+10e_1+44e_2,67h,1-67h.$

    \item[\textbf{Case 6:}]  If $k=11t+5,$ then $p=100.$ The inverse of $2p=79$ is $72.$\\
    Solving $\theta^2=p=100,$ we get $\theta = \pm 10.$\\
    Thus, $a=72(p+1)=12, \ a'=72(p-1)=110$ and $b,c=72(1\pm \theta)=78,66.$\\
    Hence, idempotent generators of QR codes are $12+78e_1+66e_2,12+66e_1+78e_2,110+55e_1+43e_2,110+43e_1+55e_2,23h,1-23h.$

    \item[\textbf{Case 7:}]  If $k=11t+6,$ then $p=23.$ The inverse of $2p=46$ is $50.$\\
    Solving $\theta^2=p=23=144,$ we get $\theta = \pm 12.$\\
    Thus, $a=50(p+1)=111, \ a'=50(p-1)=11$ and $b,c=50(1\pm \theta)=55,45.$\\
    Hence, idempotent generators of QR codes are $111+55e_1+45e_2,111+45e_1+55e_2,11+76e_1+66e_2,11+66e_1+76e_2,100h,1-100h.$

    \item[\textbf{Case 8:}]  If $k=11t+7,$ then $p=67.$ The inverse of $2p=13$ is $28.$\\
    Solving $\theta^2=p=67=1156,$ we get $\theta = \pm 34.$\\
    Thus, $a=28(p+1)=89, \ a'=28(p-1)=33$ and $b,c=28(1\pm \theta)=44,12.$\\
    Hence, idempotent generators of QR codes are $89+44e_1+12e_2,89+12e_1+44e_2,33+109e_1+77e_2,33+77e_1+109e_2,56h,1-56h.$

    \item[\textbf{Case 9:}]  If $k=11t+8,$ then $p=111.$ The inverse of $2p=101$ is $6.$\\
    Solving $\theta^2=p=111=3136,$ we get $\theta = \pm 56.$\\
    Thus, $a=6(p+1)=67, \ a'=6(p-1)=55$ and $b,c=6(1\pm \theta)=33,100.$\\
    Hence, idempotent generators of QR codes are $67+33e_1+100e_2,67+100e_1+33e_2,55+21e_1+88e_2,55+88e_1+21e_2,12h,1-12h.$

    \item[\textbf{Case 10:}]  If $k=11t+9,$ then $p=34.$ The inverse of $2p=68$ is $105.$\\
    Solving $\theta^2=p=34=1849,$ we get $\theta = \pm 43.$\\
    Thus, $a=105(p+1)=45, \ a'=105(p-1)=77$ and $b,c=105(1\pm \theta)=67,22.$\\
    Hence, idempotent generators of QR codes are $45+67e_1+22e_2,45+22e_1+67e_2,77+99e_1+54e_2,77+54e_1+99e_2,89h,1-89h.$

    \item[\textbf{Case 11:}]  If $k=11t+10,$ then $p=78.$ The inverse of $2p=35$ is $83.$\\
    Solving $\theta^2=p=78=441,$ we get $\theta = \pm 21.$\\
    Thus, $a=83(p+1)=23, \ a'=83(p-1)=99$ and $b,c=83(1\pm \theta)=34,11.$\\
    Hence, idempotent generators of QR codes are $23+34e_1+11e_2,23+11e_1+34e_2,99+110e_1+87e_2,99+87e_1+110e_2,45h,1-45h.$

    \end{enumerate}
\end{proof}

\begin{theorem}
    Let $p$ be a prime of the form $44k+5,$ then the collection of idempotents in $\displaystyle \frac{\Z_{121}[x]}{<x^p-1>}$ corresponding to each case depending on $k$ is given as follows:
    \begin{enumerate}[label=(\roman*)]
    \item If $k=11t,$ then $\{49+80e_1+17e_2,49+17e_1+80e_2,73+104e_1+41e_2,73+41e_1+104e_2,97h,1-97h\}$ are idempotents.
    \item If $k=11t+1,$ then $\{82+116e_1+47e_2, 82+47e_1+116e_2, 40+74e_1+5e_2, 40+5e_1+74e_2,42h,1-42h\}$ are idempotents.
    \item If $k=11t+2,$ then $\{115+14e_1+94e_2, 115+94e_1+14e_2, 7+27e_1+107e_2, 7+107e_1+27e_2, 108h, 1-108h\}$ are idempotents.
    \item If $k=11t+3,$ then $\{27+102e_1+72e_2, 27+72e_1+102e_2,95+49e_1+19e_2,95+19e_1+49e_2,53h,1-53h\}$ are idempotents.
    \item If $k=11t+4,$ then $\{60+50e_1+69e_2,60+69e_1+50e_2,62+52e_1+71e_2,62+71e_1+52e_2,119h,1-119h\}$ are idempotents.
    \item If $k=11t+5,$ then $\{93+36e_1+28e_2,93+28e_1+36e_2,29+93e_1+85e_2,29+85e_1+93e_2,64h,1-64h\}$ are idempotents.
    \item If $k=11t+6,$ then $\{5+6e_1+3e_2,5+3e_1+6e_2,117+1186e_1+115e_2,117+115e_1+118e_2,9h,1-9h\}$ are idempotents.
    \item If $k=11t+7,$ then $\{38+91e_1+105e_2,38+105e_1+91e_2,84+16e_1+30e_2,84+30e_1+16e_2,75h,1-75h\}$ are idempotents.
    \item If $k=11t+8,$ then $\{71+83e_1+58e_2,71+58e_1+83e_2,51+63e_1+38e_2,51+38e_1+63e_2,20h,1-20h\}$ are idempotents.
    \item If $k=11t+9,$ then $\{104+25e_1+61e_2,104+61e_1+25e_2,18+60e_1+96e_2,18+96e_1+60e_2,86h,1-86h\}$ are idempotents.
    \item If $k=11t+10,$ then $\{16+39e_1+113e_2,16+113e_1+39e_2,106+8e_1+82e_2,106+82e_1+8e_2,31h,1-31h\}$ are idempotents.
    \end{enumerate}
\end{theorem}
\begin{proof}
    Let $p=44k+5.$ We compute everything modulo $121$ and prove only the case where $k=11t+1.$ The remaining cases can be proved similarly. \\
    If $k=11t+1,$ then $p=49.$ The inverse of $2p=98$ is $21.$\\
    Solving $\theta^2=p=49,$ we get $\theta = \pm 7.$\\
    Thus, $a=21(p+1)=82, \ a'=21(p-1)=40$ and $b,c=21(1\pm \theta)=116,47.$\\
    Hence, idempotent generators of QR codes are $82+116e_1+47e_2, 82+47e_1+116e_2, 40+74e_1+5e_2, 40+5e_1+74e_2,42h,1-42h.$
\end{proof}

\begin{theorem}
     Let $p$ be a prime of the form $44k+7,$ then the collection of idempotents in $\displaystyle \frac{\Z_{121}[x]}{<x^p-1>}$ corresponding to each case depending on $k$ is given as follows:
    \begin{enumerate}[label=(\roman*)]
    \item If $k=11t,$ then $\{87+67e_1+106e_2,87+106e_1+67e_2,35+15e_1+54e_2,35+54e_1+15e_2,52h,1-52h\}$ are idempotents.
    \item If $k=11t+1,$ then $\{10+34e_1+106e_2, 10+106e_1+34e_2, 112+15e_1+87e_2, 112+87e_1+15e_2,19h,1-19h\}$ are idempotents.
    \item If $k=11t+2,$ then $\{54+e_1+106e_2, 54+106e_1+e_2, 68+15e_1-e_2, 68-e_1+15e_2, 107h, 1-107h\}$ are idempotents.
    \item If $k=11t+3,$ then $\{98+106e_1+89e_2, 98+89e_1+106e_2,24+32e_1+15e_2,24+15e_1+32e_2,74h,1-74h\}$ are idempotents.
    \item If $k=11t+4,$ then $\{21+106e_1+56e_2,21+56e_1+106e_2,101+65e_1+15e_2,101+15e_1+65e_2,41h,1-41h\}$ are idempotents.
    \item If $k=11t+5,$ then $\{65+106e_1+23e_2,65+23e_1+106e_2,57+98e_1+15e_2,57+15e_1+98e_2,8h,1-8h\}$ are idempotents.
    \item If $k=11t+6,$ then $\{109+106e_1+111e_2,109+111e_1+106e_2,13+10e_1+15e_2,13+15e_1+10e_2,96h,1-96h\}$ are idempotents.
    \item If $k=11t+7,$ then $\{32+106e_1+78e_2,32+78e_1+106e_2,90+43e_1+15e_2,90+15e_1+43e_2,63h,1-63h\}$ are idempotents.
    \item If $k=11t+8,$ then $\{76+106e_1+45e_2,76+45e_1+106e_2,46+76e_1+15e_2,46+15e_1+76e_2,30h,1-30h\}$ are idempotents.
    \item If $k=11t+9,$ then $\{-1+12e_1+106e_2,-1+106e_1+12e_2,2+15e_1+109e_2,2+109e_1+15e_2,118h,1-118h\}$ are idempotents.
    \item If $k=11t+10,$ then $\{43+100e_1+106e_2,43+106e_1+100e_2,79+15e_1+21e_2,79+21e_1+15e_2,85h,1-85h\}$ are idempotents.
    \end{enumerate}
\end{theorem}
\begin{proof}
    Let $p=44k+7.$ We compute everything modulo $121$ and prove only the case where $k=11t+2.$ The remaining cases can be proved similarly. \\
    If $k=11t+2,$ then $p=95.$ The inverse of $2p=69$ is $114.$\\
    Solving $\theta^2=-p=-95=2809,$ we get $\theta = \pm 53.$\\
    Thus, $a=114(p+1)=54, \ a'=114(p-1)=68$ and $b,c=114(1\pm \theta)=1,106.$\\
    Hence, idempotent generators of QR codes are $54+e_1+106e_2, 54+106e_1+e_2, 68+15e_1-e_2, 68-e_1+15e_2, 107h, 1-107h.$
\end{proof}

\begin{theorem}
    Let $p$ be a prime of the form $44k+9,$ then the collection of idempotents in $\displaystyle \frac{\Z_{121}[x]}{<x^p-1>}$ corresponding to each case depending on $k$ is given as follows:
    \begin{enumerate}[label=(\roman*)]
    \item If $k=11t,$ then $\{14+94e_1+54e_2,14+54e_1+94e_2,108+67e_1+27e_2,108+27e_1+67e_2,27h,1-27h\}$ are idempotents.
    \item If $k=11t+1,$ then $\{69+10e_1+6e_2, 69+6e_1+10e_2, 53-6e_1-10e_2, 53-10e_1-6e_2,16h,1-16h\}$ are idempotents.
    \item If $k=11t+2,$ then $\{3+39e_1+87e_2, 3+87e_1+39e_2, -2+34e_1+82e_2, -2+82e_1+34e_2, 5h, 1-5h\}$ are idempotents.
    \item If $k=11t+3,$ then $\{58+72e_1+43e_2, 58+43e_1+72e_2,64+78e_1+49e_2,64+49e_1+78e_2,115h,1-115h\}$ are idempotents.
    \item If $k=11t+4,$ then $\{113-e_1+105e_2,113+105e_1-e_2,9+16e_1+e_2,9+e_1+16e_2,104h,1-104h\}$ are idempotents.
    \item If $k=11t+5,$ then $\{47+76e_1+17e_2,47+17e_1+76e_2,75+104e_1+45e_2,75+45e_1+104e_2,93h,1-93h\}$ are idempotents.
    \item If $k=11t+6,$ then $\{102+50e_1+32e_2,102+32e_1+50e_2,20+89e_1+71e_2,20+71e_1+89e_2,82h,1-82h\}$ are idempotents.
    \item If $k=11t+7,$ then $\{36+83e_1+109e_2,36+109e_1+83e_2,86+12e_1+38e_2,86+38e_1+12e_2,71h,1-71h\}$ are idempotents.
    \item If $k=11t+8,$ then $\{91+65e_1+116e_2,91+116e_1+65e_2,31+5e_1+56e_2,31+56e_1+5e_2,60h,1-60h\}$ are idempotents.
    \item If $k=11t+9,$ then $\{25+21e_1+28e_2,25+28e_1+21e_2,97+93e_1+100e_2,97+100e_1+93e_2,49h,1-49h\}$ are idempotents.
    \item If $k=11t+10,$ then $\{80+61e_1+98e_2,80+98e_1+61e_2,42+23e_1+60e_2,42+60e_1+23e_2,38h,1-38h\}$ are idempotents.
    \end{enumerate}
\end{theorem}
\begin{proof}
    Let $p=44k+9.$ We compute everything modulo $121$ and prove only the case where $k=11t+3.$ The remaining cases can be proved similarly. \\
    If $k=11t+3,$ then $p=20.$ The inverse of $2p=40$ is $-3.$\\
    Solving $\theta^2=p=20=625,$ we get $\theta = \pm 25.$\\
    Thus, $a=-3(p+1)=58, \ a'=-3(p-1)=64$ and $b,c=-3(1\pm \theta)=72,43.$\\
    Hence, idempotent generators of QR codes are $58+72e_1+43e_2, 58+43e_1+72e_2,64+78e_1+49e_2,64+49e_1+78e_2,115h,1-115h.$
\end{proof}

\begin{theorem}
     Let $p$ be a prime of the form $44k+19,$ then the collection of idempotents in $\displaystyle \frac{\Z_{121}[x]}{<x^p-1>}$ corresponding to each case depending on $k$ is given as follows:
    \begin{enumerate}[label=(\roman*)]
    \item If $k=11t,$ then $\{26+107e_1+65e_2,26+65e_1+107e_2,96+56e_1+14e_2,96+14e_1+56e_2,51h,1-51h\}$ are idempotents.
    \item If $k=11t+1,$ then $\{37+43e_1+30e_2, 37+30e_1+43e_2, 85+91e_1+78e_2, 85+78e_1+91e_2,73h,1-73h\}$ are idempotents.
    \item If $k=11t+2,$ then $\{48+74e_1+21e_2, 48+21e_1+74e_2, 74+100e_1+47e_2, 74+47e_1+100e_2, 95h, 1-95h\}$ are idempotents.
    \item If $k=11t+3,$ then $\{59-3e_1-e_2, 59-e_1-3e_2,63+e_1+3e_2,63+3e_1+e_2,117h,1-117h\}$ are idempotents.
    \item If $k=11t+4,$ then $\{70+98e_1+41e_2,70+41e_1+98e_2,52+80e_1+23e_2,52+23e_1+80e_2,18h,1-18h\}$ are idempotents.
    \item If $k=11t+5,$ then $\{81+85e_1+76e_2,81+76e_1+85e_2,41+45e_1+36e_2,41+36e_1+45e_2,40h,1-40h\}$ are idempotents.
    \item If $k=11t+6,$ then $\{92+54e_1+8e_2,92+8e_1+54e_2,30+113e_1+67e_2,30+67e_1+113e_2,62h,1-62h\}$ are idempotents.
    \item If $k=11t+7,$ then $\{103+32e_1+52e_2,103+52e_1+32e_2,19+69e_1+89e_2,19+89e_1+69e_2,84h,1-84h\}$ are idempotents.
    \item If $k=11t+8,$ then $\{114+96e_1+10e_2,114+10e_1+96e_2,8+111e_1+25e_2,8+25e_1+111e_2,106h,1-106h\}$ are idempotents.
    \item If $k=11t+9,$ then $\{4+109e_1+19e_2,4+19e_1+109e_2,-3+102e_1+12e_2,-3+12e_1+102e_2,7h,1-7h\}$ are idempotents.
    \item If $k=11t+10,$ then $\{15+63e_1+87e_2,15+87e_1+63e_2,107+34e_1+58e_2,107+58e_1+34e_2,29h,1-29h\}$ are idempotents.
    \end{enumerate}
\end{theorem}
\begin{proof}
    Let $p=44k+19.$ We compute everything modulo $121$ and prove only the case where $k=11t+4.$ The remaining cases can be proved similarly. \\
    If $k=11t+4,$ then $p=74.$ The inverse of $2p=27$ is $9.$\\
    Solving $\theta^2=-p=-74=289,$ we get $\theta = \pm 17.$\\
    Thus, $a=9(p+1)=70, \ a'=9(p-1)=52$ and $b,c=9(1\pm \theta)=98,41.$\\
    Hence, idempotent generators of QR codes are $70+98e_1+41e_2,70+41e_1+98e_2,52+80e_1+23e_2,52+23e_1+80e_2,18h,1-18h.$
\end{proof}

\begin{theorem}
     Let $p$ be a prime of the form $44k-19,$ then the collection of idempotents in $\displaystyle \frac{\Z_{121}[x]}{<x^p-1>}$ corresponding to each case depending on $k$ is given as follows:
    \begin{enumerate}[label=(\roman*)]
    \item If $k=11t,$ then $\{107+58e_1+34e_2,107+34e_1+58e_2,15+87e_1+63e_2,15+63e_1+87e_2,92h,1-92h\}$ are idempotents.
    \item If $k=11t+1,$ then $\{-3+12e_1+102e_2, -3+102e_1+12e_2, 4+19e_1+109e_2, 4+109e_1+19e_2,114h,1-114h\}$ are idempotents.
    \item If $k=11t+2,$ then $\{8+25e_1+111e_2, 8+111e_1+25e_2, 114+10e_1+96e_2, 114+96e_1+10e_2, 15h, 1-15h\}$ are idempotents.
    \item If $k=11t+3,$ then $\{19+89e_1+69e_2, 19+69e_1+89e_2,103+52e_1+32e_2,103+32e_1+52e_2,37h,1-37h\}$ are idempotents.
    \item If $k=11t+4,$ then $\{30+67e_1+113e_2,30+113e_1+67e_2,92+8e_1+54e_2,92+54e_1+8e_2,59h,1-59h\}$ are idempotents.
    \item If $k=11t+5,$ then $\{41+36e_1+45e_2,41+45e_1+36e_2,81+76e_1+85e_2,81+85e_1+76e_2,81h,1-81h\}$ are idempotents.
    \item If $k=11t+6,$ then $\{52+23e_1+80e_2,52+80e_1+23e_2,70+41e_1+98e_2,70+98e_1+41e_2,103h,1-103h\}$ are idempotents.
    \item If $k=11t+7,$ then $\{63+3e_1+e_2,63+e_1+3e_2,59-e_1-3e_2,59-3e_1-e_2,4h,1-4h\}$ are idempotents.
    \item If $k=11t+8,$ then $\{74+47e_1+100e_2,74+100e_1+47e_2,48+21e_1+74e_2,48+74e_1+21e_2,26h,1-26h\}$ are idempotents.
    \item If $k=11t+9,$ then $\{85+78e_1+91e_2,85+91e_1+78e_2,37+30e_1+43e_2,37+43e_1+30e_2,48h,1-48h\}$ are idempotents.
    \item If $k=11t+10,$ then $\{96+14e_1+56e_2,96+56e_1+14e_2,26+65e_1+107e_2,26+107e_1+65e_2,70h,1-70h\}$ are idempotents.
    \end{enumerate}
\end{theorem}
\begin{proof}
    Let $p=44k-19.$ We compute everything modulo $121$ and prove only the case where $k=11t+5.$ The remaining cases can be proved similarly. \\
    If $k=11t+5,$ then $p=3.$ The inverse of $2p=6$ is $101.$\\
    Solving $\theta^2=p=3=729,$ we get $\theta = \pm 27.$\\
    Thus, $a=101(p+1)=41, \ a'=101(p-1)=81$ and $b,c=101(1\pm \theta)=36,45.$\\
    Hence, idempotent generators of QR codes are $41+36e_1+45e_2,41+45e_1+36e_2,81+76e_1+85e_2,81+85e_1+76e_2,81h,1-81h.$
\end{proof}

\begin{theorem}
     Let $p$ be a prime of the form $44k-9,$ then the collection of idempotents in $\displaystyle \frac{\Z_{121}[x]}{<x^p-1>}$ corresponding to each case depending on $k$ is given as follows:
    \begin{enumerate}[label=(\roman*)]
    \item If $k=11t,$ then $\{42+60e_1+23e_2,42+23e_1+60e_2,80+98e_1+61e_2,80+61e_1+98e_2,83h,1-83h\}$ are idempotents.
    \item If $k=11t+1,$ then $\{97+100e_1+93e_2, 97+93e_1+100e_2, 25+28e_1+21e_2, 25+21e_1+28e_2,72h,1-72h\}$ are idempotents.
    \item If $k=11t+2,$ then $\{31+56e_1+5e_2, 31+5e_1+56e_2, 91+116e_1+65e_2, 91+65e_1+116e_2, 61h, 1-61h\}$ are idempotents.
    \item If $k=11t+3,$ then $\{86+38e_1+12e_2, 86+12e_1+38e_2,36+109e_1+83e_2,36+83e_1+109e_2,50h,1-50h\}$ are idempotents.
    \item If $k=11t+4,$ then $\{20+71e_1+89e_2,20+89e_1+71e_2,102+32e_1+50e_2,102+50e_1+32e_2,39h,1-39h\}$ are idempotents.
    \item If $k=11t+5,$ then $\{75+45e_1+104e_2,75+104e_1+45e_2,47+17e_1+76e_2,47+76e_1+17e_2,28h,1-28h\}$ are idempotents.
    \item If $k=11t+6,$ then $\{9+e_1+16e_2,9+16e_1+e_2,113+105e_1-e_2,113-e_1+105e_2,17h,1-17h\}$ are idempotents.
    \item If $k=11t+7,$ then $\{64+49e_1+78e_2,64+78e_1+49e_2,58+43e_1+72e_2,58+72e_1+43e_2,6h,1-6h\}$ are idempotents.
    \item If $k=11t+8,$ then $\{-2+82e_1+34e_2,-2+34e_1+82e_2,3+87e_1+39e_2,3+39e_1+87e_2,116h,1-116h\}$ are idempotents.
    \item If $k=11t+9,$ then $\{53+111e_1+115e_2,53+115e_1+111e_2,69+6e_1+10e_2,69+10e_1+6e_2,105h,1-105h\}$ are idempotents.
    \item If $k=11t+10,$ then $\{108+27e_1+67e_2,108+67e_1+27e_2,14+54e_1+94e_2,14+94e_1+54e_2,94h,1-94h\}$ are idempotents.
    \end{enumerate}
\end{theorem}
\begin{proof}
    Let $p=44k-9.$ We compute everything modulo $121$ and prove only the case where $k=11t+6.$ The remaining cases can be proved similarly. \\
    If $k=11t+6,$ then $p=57.$ The inverse of $2p=114$ is $69.$\\
    Solving $\theta^2=-p=-57=64,$ we get $\theta = \pm 8.$\\
    Thus, $a=69(p+1)=9, \ a'=69(p-1)=113$ and $b,c=69(1\pm \theta)=1,16.$\\
    Hence, idempotent generators of QR codes are $9+e_1+16e_2,9+16e_1+e_2,113+105e_1-e_2,113-e_1+105e_2,17h,1-17h.$
\end{proof}

\begin{theorem}
     Let $p$ be a prime of the form $44k-7,$ then the collection of idempotents in $\displaystyle \frac{\Z_{121}[x]}{<x^p-1>}$ corresponding to each case depending on $k$ is given as follows:
    \begin{enumerate}[label=(\roman*)]
    \item If $k=11t,$ then $\{79+21e_1+15e_2,79+15e_1+21e_2,43+106e_1+100e_2,43+100e_1+106e_2,36h,1-36h\}$ are idempotents.
    \item If $k=11t+1,$ then $\{2+109e_1+15e_2, 2+15e_1+109e_2, -1+106e_1+12e_2, -1+12e_1+106e_2,3h,1-3h\}$ are idempotents.
    \item If $k=11t+2,$ then $\{46+15e_1+76e_2, 46+76e_1+15e_2, 76+45e_1+106e_2, 76+106e_1+45e_2, 91h, 1-91h\}$ are idempotents.
    \item If $k=11t+3,$ then $\{90+15e_1+43e_2, 90+43e_1+15e_2,32+78e_1+106e_2,32+106e_1+78e_2,58h,1-58h\}$ are idempotents.
    \item If $k=11t+4,$ then $\{13+15e_1+10e_2,13+10e_1+15e_2,109+111e_1+106e_2,109+106e_1+111e_2,25h,1-25h\}$ are idempotents.
    \item If $k=11t+5,$ then $\{57+15e_1+98e_2,57+98e_1+15e_2,65+23e_1+106e_2,65+106e_1+23e_2,113h,1-113h\}$ are idempotents.
    \item If $k=11t+6,$ then $\{101+15e_1+65e_2,101+65e_1+15e_2,21+56e_1+106e_2,21+106e_1+56e_2,80h,1-80h\}$ are idempotents.
    \item If $k=11t+7,$ then $\{24+15e_1+32e_2,24+32e_1+15e_2,98+89e_1+106e_2,98+106e_1+89e_2,47h,1-47h\}$ are idempotents.
    \item If $k=11t+8,$ then $\{68-e_1+15e_2,68+15e_1-e_2,54+106e_1+e_2,54+e_1+106e_2,14h,1-14h\}$ are idempotents.
    \item If $k=11t+9,$ then $\{112+87e_1+15e_2,112+15e_1+87e_2,10+106e_1+34e_2,10+34e_1+106e_2,102h,1-102h\}$ are idempotents.
    \item If $k=11t+10,$ then $\{35+54e_1+15e_2,35+15e_1+54e_2,87+106e_1+67e_2,87+67e_1+106e_2,69h,1-69h\}$ are idempotents.
    \end{enumerate}
\end{theorem}
\begin{proof}
    Let $p=44k-7.$ We compute everything modulo $121$ and prove only the case where $k=11t+7.$ The remaining cases can be proved similarly. \\
    If $k=11t+7,$ then $p=103.$ The inverse of $2p=85$ is $84.$\\
    Solving $\theta^2=p=103=3249,$ we get $\theta = \pm 57.$\\
    Thus, $a=84(p+1)=24, \ a'=84(p-1)=98$ and $b,c=84(1\pm \theta)=15,32.$\\
    Hence, idempotent generators of QR codes are $24+15e_1+32e_2,24+32e_1+15e_2,98+89e_1+106e_2,98+106e_1+89e_2,47h,1-47h.$
\end{proof}

\begin{theorem}
     Let $p$ be a prime of the form $44k-5,$ then the collection of idempotents in $\displaystyle \frac{\Z_{121}[x]}{<x^p-1>}$ corresponding to each case depending on $k$ is given as follows:
    \begin{enumerate}[label=(\roman*)]
    \item If $k=11t,$ then $\{106+82e_1+8e_2,106+8e_1+82e_2,16+113e_1+39e_2,16+39e_1+113e_2,90h,1-90h\}$ are idempotents.
    \item If $k=11t+1,$ then $\{18+96e_1+60e_2, 18+60e_1+96e_2, 104+61e_1+25e_2, 104+25e_1+61e_2,35h,1-35h\}$ are idempotents.
    \item If $k=11t+2,$ then $\{51+38e_1+63e_2, 51+63e_1+38e_2, 71+58e_1+83e_2, 71+83e_1+58e_2, 101h, 1-101h\}$ are idempotents.
    \item If $k=11t+3,$ then $\{84+30e_1+16e_2, 84+16e_1+30e_2,38+105e_1+91e_2,38+91e_1+105e_2,46h,1-46h\}$ are idempotents.
    \item If $k=11t+4,$ then $\{-4-6e_1-3e_2,-4-3e_1-6e_2,5+3e_1+6e_2,5+6e_1+3e_2,112h,1-112h\}$ are idempotents.
    \item If $k=11t+5,$ then $\{29+85e_1+93e_2,29+93e_1+85e_2,93+28e_1+36e_2,93+36e_1+28e_2,57h,1-57h\}$ are idempotents.
    \item If $k=11t+6,$ then $\{62+71e_1+52e_2,62+52e_1+71e_2,60+69e_1+50e_2,60+50e_1+69e_2,2h,1-2h\}$ are idempotents.
    \item If $k=11t+7,$ then $\{95+19e_1+49e_2,95+49e_1+19e_2,27+72e_1+102e_2,27+102e_1+72e_2,68h,1-68h\}$ are idempotents.
    \item If $k=11t+8,$ then $\{7+107e_1+27e_2,7+27e_1+107e_2,115+94e_1+14e_2,115+14e_1+94e_2,13h,1-13h\}$ are idempotents.
    \item If $k=11t+9,$ then $\{40+5e_1+74e_2,40+74e_1+5e_2,82+47e_1+116e_2,82+116e_1+47e_2,79h,1-79h\}$ are idempotents.
    \item If $k=11t+10,$ then $\{73+41e_1+104e_2,73+104e_1+41e_2,49+17e_1+80e_2,49+80e_1+17e_2,24h,1-24h\}$ are idempotents.
    \end{enumerate}
\end{theorem}
\begin{proof}
    Let $p=44k-5.$ We compute everything modulo $121$ and prove only the case where $k=11t+8.$ The remaining cases can be proved similarly. \\
    If $k=11t+8,$ then $p=28.$ The inverse of $2p=56$ is $67.$\\
    Solving $\theta^2=-p=-28=3481,$ we get $\theta = \pm 59.$\\
    Thus, $a=67(p+1)=7, \ a'=67(p-1)=115$ and $b,c=67(1\pm \theta)=107,27.$\\
    Hence, idempotent generators of QR codes are $7+107e_1+27e_2,7+27e_1+107e_2,115+94e_1+14e_2,115+14e_1+94e_2,13h,1-13h.$
\end{proof}

\begin{theorem}
     Let $p$ be a prime of the form $44k-1,$ then the collection of idempotents in $\displaystyle \frac{\Z_{121}[x]}{<x^p-1>}$ corresponding to each case depending on $k$ is given as follows:
    \begin{enumerate}[label=(\roman*)]
    \item If $k=11t,$ then $\{99+87e_1+110e_2,99+110e_1+87e_2,23+11e_1+34e_2,23+34e_1+11e_2,76h,1-76h\}$ are idempotents.
    \item If $k=11t+1,$ then $\{77+54e_1+99e_2, 77+99e_1+54e_2, 45+22e_1+67e_2, 45+67e_1+22e_2,32h,1-32h\}$ are idempotents.
    \item If $k=11t+2,$ then $\{55+88e_1+21e_2, 55+21e_1+88e_2, 67+100e_1+33e_2, 67+33e_1+100e_2, 109h, 1-109h\}$ are idempotents.
    \item If $k=11t+3,$ then $\{33+77e_1+109e_2, 33+109e_1+77e_2,89+12e_1+44e_2,89+44e_1+12e_2,65h,1-65h\}$ are idempotents.
    \item If $k=11t+4,$ then $\{11+66e_1+76e_2,11+76e_1+66e_2,111+45e_1+55e_2,111+55e_1+45e_2,21h,1-21h\}$ are idempotents.
    \item If $k=11t+5,$ then $\{110+43e_1+55e_2,110+55e_1+43e_2,12+66e_1+78e_2,12+78e_1+66e_2,98h,1-98h\}$ are idempotents.
    \item If $k=11t+6,$ then $\{88+10e_1+44e_2,88+44e_1+10e_2,34+77e_1+111e_2,34+111e_1+77e_2,54h,1-54h\}$ are idempotents.
    \item If $k=11t+7,$ then $\{66+98e_1+33e_2,66+33e_1+98e_2,56+88e_1+23e_2,56+23e_1+88e_2,10h,1-10h\}$ are idempotents.
    \item If $k=11t+8,$ then $\{44+22e_1+65e_2,44+65e_1+22e_2,78+56e_1+99e_2,78+99e_1+56e_2,87h,1-87h\}$ are idempotents.
    \item If $k=11t+9,$ then $\{22+11e_1+32e_2,22+32e_1+11e_2,100+89e_1+110e_2,100+110e_1+89e_2,43h,1-43h\}$ are idempotents.
    \item If $k=11t+10,$ then $\{-e_2,-e_1,1+e_1,1+e_2,-h,1+h\}$ are idempotents.
    \end{enumerate}
\end{theorem}
\begin{proof}
    Let $p=44k-1.$ We compute everything modulo $121$ and prove only the case where $k=11t+9.$ The remaining cases can be proved similarly. \\
    If $k=11t+9,$ then $p=76.$ The inverse of $2p=31$ is $82.$\\
    Solving $\theta^2=-p=-76=529,$ we get $\theta = \pm 23.$\\
    Thus, $a=82(p+1)=22, \ a'=82(p-1)=100$ and $b,c=82(1\pm \theta)=11,32.$\\
    Hence, idempotent generators of QR codes are $22+11e_1+32e_2,22+32e_1+11e_2,100+89e_1+110e_2,100+110e_1+89e_2,43h,1-43h.$
\end{proof}

\subsection{\textbf{Properties of QR codes over \texorpdfstring{\ensuremath{\Z_{121}}}{Z121}}}
We provide some necessary definitions before moving on to the properties of quadratic residue codes.

\begin{definition}
    Let $x=(x_1,x_2,\dots,x_n)$ and $y=(y_1,y_2,\dots,y_n)$ be two vectors of $\Z_{121}^n.$ The inner product of $x$ and $y$ is defined as:
    \[
    <x,y>=x_1y_1+x_2y_2+\cdots +x_ny_n \ ,
    \]
    where the operation is performed in $\Z_{121}.$ 
\begin{itemize}
    \item The dual code $C^\perp$ of $C$ is defined as:
    \[
    C^\perp =\{x \in \Z_{121}^n \, | \, <x,c>=0, \, \forall \ c \in C \}.
    \]

    \item The Code $C$ is said to be self-orthogonal if $C\subseteq C^\perp$ and self-dual if $C=C^\perp.$
\end{itemize}
\end{definition}

If a $\Z_{121}$- cyclic code is generated by one of the idempotents in the last ten theorems, then it is called a $\mathbb{Z}_{121}$- QR code. Let $a \in \Z_p^*,$ then the multiplier function $\mu _a$ defined by \[\mu _a (i) \equiv ia\, (\bmod \, p)\] is a permutation of the coordinate positions $\{0,1,2,\dots,p-1\}$ of a cyclic code of length $p.$ Since cyclic codes of length $p$ over $\Z_{121}$ are the ideals in $\displaystyle \frac{\Z_{121}[x]}{<x^p-1>},$ it is clear that for any polynomial $\displaystyle f,g \in \frac{\Z_{121}[x]}{<x^p-1>},$ we have
\[ 
\mu _a (fg) = \mu _a (f) \mu _a (g).
\]

We investigate several properties of QR codes depending on $p.$ Also, throughout the section, $Q_1, Q_2$ will denote the QR codes generated by $r(x)$ and $n(x)$, respectively, and $Q_1', Q_2'$ will denote the QR codes generated by $(x-1).r(x)$ and $(x-1).n(x),$ respectively.

\begin{theorem}\label{thm:5.13}
    Let $p$ be a prime of the form $44k+1.$ Then the collection of ideals in $\displaystyle \frac{\Z_{121}[x]}{<x^p-1>}$ corresponding to each case depending on $k$ is given as follows:
    \begin{enumerate}[label=(\roman*)]
    \item If $k=11t,$ let $Q_1=<1+e_2>, Q_2=<1+e_1>, Q'_1=<-e_1>, \ Q'_2=<-e_2>.$
    
    \item If $k=11t+1,$ let $Q_1=<100+110e_1+89e_2>,\ Q_2=<100+89e_1+110e_2>,\\
    Q'_1=<22+32e_1+11e_2>,\ Q'_2=<22+11e_1+32e_2>.$

    \item If $k=11t+2,$ let $ Q_1=<78+99e_1+56e_2>, \ Q_2=<78+56e_1+99e_2>,\\
    Q'_1=<44+65e_1+22e_2>, \ Q'_2=<44+22e_1+65e_2>.$
    
    \item If $k=11t+3,$ let $ Q_1=<56+23e_1+88e_2>, \ Q_2=<56+88e_1+23e_2>,\\
    Q'_1=<66+33e_1+98e_2>, \ Q'_2=<66+98e_1+33e_2>.$
   
    \item If $k=11t+4,$ let $Q_1=<34+111e_1+77e_2>, \ Q_2=<34+77e_1+111e_2>,\\
    Q'_1=<88+44e_1+10e_2>, \ Q'_2=<88+10e_1+44e_2>.$
    
    \item If $k=11t+5,$ let $Q_1=<12+78e_1+66e_2>, \ Q_2=<12+66e_1+78e_2>,\\
    Q'_1=<110+55e_1+43e_2>, \ Q'_2=<110+43e_1+55e_2>.$
   
    \item If $k=11t+6,$ let $Q_1=<111+55e_1+45e_2>, \ Q_2=<111+45e_1+55e_2>,\\
    Q'_1=<11+76e_1+66e_2>, \ Q'_2=<11+66e_1+76e_2>.$

    \item If $k=11t+7,$ let $Q_1=<89+44e_1+12e_2>, \ Q_2=<89+12e_1+44e_2>,\\
    Q'_1=<33+109e_1+77e_2>, \ Q'_2=<33+77e_1+109e_2>. $

    \item If $k=11t+8,$ let $Q_1=<67+33e_1+100e_2>, \ Q_2=<67+100e_1+33e_2>,\\
    Q'_1=<55+21e_1+88e_2>, \ Q'_2=<55+88e_1+21e_2>.$

    \item If $k=11t+9,$ let $Q_1=<45+67e_1+22e_2>, \ Q_2=<45+22e_1+67e_2>,\\
    Q'_1=<77+99e_1+54e_2>, \ Q'_2=<77+54e_1+99e_2>.$

    \item If $k=11t+10,$ let $Q_1=<23+34e_1+11e_2>, \ Q_2=<23+11e_1+34e_2>,\\
    Q'_1=<99+110e_1+87e_2>, \ Q'_2=<99+87e_1+110e_2>.$
    \end{enumerate}
    Then the following holds for $\Z_{121}$- QR codes $Q_1,Q_2,Q'_1$ and $Q'_2.$

    \begin{enumerate}[label=(\roman*)]
        \item[(a)] $Q_1$ is equivalent to $Q_2$ and $Q'_1$ is equivalent to $Q'_2.$
        \item[(b)] $\displaystyle Q_1 \cap Q_2=<\bar{h}>, \, Q_1+Q_2=\frac{\Z_{121}[x]}{<x^p-1>}, \text{where} \\ \bar{h}\in \{h,12h,23h,34h,45h,56h,67h,78h,89h,100h,111h\}.$
        \item[(c)] $Q_1=Q'_1+<\bar{h}>,$ $Q_2=Q'_2+<\bar{h}>.$
        \item[(d)] $|Q_1|=121^{\frac{p+1}{2}}=|Q_2|,$ $|Q'_1|=121^{\frac{p-1}{2}}=|Q'_2|.$
        \item[(e)] $Q_1^\perp = Q'_2,$ $Q_2^\perp = Q'_1.$
        \item[(f)] $Q'_1 \cap Q'_2=\{0\},$ $Q'_1 + Q'_2=< 1-\bar{h}>.$
    \end{enumerate}
\end{theorem}

\begin{proof}
   Let $p=44k+1.$ We compute everything modulo $121$and prove the case where $k=11t+1.$ The remaining cases can be proved similarly.
    \begin{enumerate}[label=(\roman*)]
        \item[(a)] Let $a \in \mathcal{N}.$ Then,
        \[
        \begin{aligned}
        \mu _a (<100+110e_1+89e_2>)=100+89e_1+110e_2 = g_2 \ (say), \\
        \mu _a (<100+89e_1+110e_2>)=100+110e_1+89e_2 = g_1 \ (say), \\
        \mu _a (<22+32e_1+11e_2>)=22+11e_1+32e_2 = g_2' \ (say), \\
        \mu _a (<22+11e_1+32e_2>)=22+32e_1+11e_2 = g_1' \ (say). \\
        \end{aligned}
        \]
        Hence, $Q_1 \sim Q_2,$ $Q'_1 \sim Q'_2.$

        \item[(b)] Using Theorem- \ref{thm:4.7}, $Q_1 \cap Q_2 = <g_1 g_2>=<(100+110e_1+89e_2)(100+89e_1+110e_2)>.$\\ 
        Since $g_1 + g_2 = (100+110e_1+89e_2)+(100+89e_1+110e_2)=79+78(e_1+e_2)=1+78h,$ we have $g_2 = (1+78h)-g_1.$\\
        Therefore, $g_1g_2=g_1(1+78h-g_1)=g_1(1+78h)-g_1=g_1(78h)=(100+110e_1+89e_2)78h=(100h+110(\frac{p-1}{2})h+89(\frac{p-1}{2})h)78=78h.$\\
        Hence, $Q_1 \cap Q_2$ has idempotent generator $78h.$
        \\
        Also, $\displaystyle Q_1+Q_2=<g_1+g_2-g_1g_2>=<1>=\frac{\Z_{121}[x]}{<x^p-1>}.$

        \item[(c)] From Theorem- \ref{thm:4.7}, $Q'_1+<\bar{h}>$ has idempotent generator 
        \[
        (22+32e_1+11e_2) + 78h - (22+32e_1+11e_2)(78h).
        \]
        Thus, $Q'_1+<\bar{h}> = (22+32e_1+11e_2) + 78h - (22+32e_1+11e_2)(78h)=100+110e_1+89e_2=Q_1.$\\
        Similarly, $Q'_2+<\bar{h}>=Q_2.$

        \item[(d)] Since, $\displaystyle 121^p=|Q_1 + Q_2|=\frac{|Q_1|Q_2||}{|Q_1 \cap Q_2|}=\frac{|Q_1|^2}{121},$ we have, 
        \[
        |Q_1|=121^{\frac{p+1}{2}}=|Q_2|.
        \]
    
        Again, since $Q'_1 + <h> = Q_1=|Q'_1||<h>|,$ therefore we have 
        \[
        |Q'_1|=|Q_1||<h>|=121^{\frac{p-1}{2}}.
        \]
        Similarly, $|Q'_2|=121^{\frac{p-1}{2}}.$ 

        \item[(e)] Since $p \equiv1 \pmod 4,$ therefore $-1 \in \mathcal{Q}.$ Also, by Theorem- \ref{thm:4.7}, $Q_1^\perp$ has idempotent generator $1-g_1(x^{-1})=1-(100+110e_1(x^{-1})+89e_2(x^{-1}))=22+11e_1(x^{-1})+32e_2(x^{-1})=Q'_2.$\\
        That is, $Q_1^\perp = Q'_2.$\\
        Similarly, $Q_2^\perp = Q'_1.$

        \item[(f)] $Q'_1 \cap Q'_2=<g'_1g'_2>=<(22+32e_1+11e_2)(22+11e_1+32e_2)>.$\\
        Since, $g'_1+g'_2=44+43(e_1+e_2)=1+43h,$ we have $g'_2=1+43h-g'_1.$\\
        Therefore, $g'_1g'_2=g'_1(1+43h-g'_1)=(43h)g'_1=0.$\\
        Hence, $Q'_1 \cap Q'_2=\{0\}.$
        
        Also, $Q'_1+Q'_2=<g'_1+g'_2-g'_1g'_2>=<1+43h>=<1-78h>=<1-\bar{h}>.$
        
    \end{enumerate}

\end{proof}

The subsequent theorems can be proven for various prime instances $p$, similar to Theorem- \ref{thm:5.13}.

\begin{theorem}\label{thm:5.14}
    Let $p$ be a prime of the form $44k+5.$ Then the collection of ideals in $\displaystyle \frac{\Z_{121}[x]}{<x^p-1>}$ corresponding to each case depending on $k$ is given as follows:
    \begin{enumerate}[label=(\roman*)]
    \item If $k=11t,$ let $Q_1=<49+80e_1+17e_2>, \ Q_2=<49+17e_1+80e_2>, \\ 
    Q'_1=<73+104e_1+41e_2>, \ Q'_2=<73+41e_1+104e_2>.$
    
    \item If $k=11t+1,$ let $Q_1=<82+116e_1+47e_2>,\ Q_2=<82+47e_1+116e_2>,\\
    Q'_1=<40+74e_1+5e_2>,\ Q'_2=<40+5e_1+74e_2>.$
    
    \item If $k=11t+2,$ let $Q_1=<115+14e_1+94e_2>, \ Q_2=<115+94e_1+14e_2>,\\
    Q'_1=<7+27e_1+107e_2>, \ Q'_2=<7+107e_1+27e_2>.$
    
    \item If $k=11t+3,$ let $Q_1=<27+102e_1+72e_2>, \ Q_2=<27+72e_1+102e_2>,\\
    Q'_1=<95+49e_1+19e_2>, \ Q'_2=<95+19e_1+49e_2>. $
    
    \item If $k=11t+4,$ let $Q_1=<60+50e_1+69e_2>, \ Q_2=<60+69e_1+50e_2>,\\
    Q'_1=<62+52e_1+71e_2>, \ Q'_2=<62+71e_1+52e_2>.$
    
    \item If $k=11t+5,$ let $Q_1=<93+36e_1+28e_2>, \ Q_2=<93+28e_1+36e_2>,\\
    Q'_1=<29+93e_1+85e_2>, \ Q'_2=<29+85e_1+93e_2>. $
    
    \item If $k=11t+6,$ let $Q_1=<5+6e_1+3e_2>, \ Q_2=<5+3e_1+6e_2>,\\
    Q'_1=<117+118e_1+115e_2>, \ Q'_2=<117+115e_1+118e_2>.$

    \item If $k=11t+7,$ let $Q_1=<38+91e_1+105e_2>, \ Q_2=<38+105e_1+91e_2>,\\
    Q'_1=<84+16e_1+30e_2>, \ Q'_2=<84+30e_1+16e_2>.$

    \item If $k=11t+8,$ let $Q_1=<71+83e_1+58e_2>, \ Q_2=<71+58e_1+83e_2>,\\
    Q'_1=<51+63e_1+38e_2>, \ Q'_2=<51+38e_1+63e_2>. $
    
    \item If $k=11t+9,$ let $Q_1=<104+25e_1+6e_2>, \ Q_2=<104+6e_1+25e_2>,\\
    Q'_1=<18+60e_1+96e_2>, \ Q'_2=<18+96e_1+60e_2>.$
    
    \item If $k=11t+10,$ let $Q_1=<16+39e_1+113e_2>, \ Q_2=<16+113e_1+39e_2>,\\
    Q'_1=<106+8e_1+82e_2>, \ Q'_2=<106+82e_1+8e_2>.$
    \end{enumerate}
    Then the following holds for $\Z_{121}$- QR codes $Q_1,Q_2,Q'_1$ and $Q'_2.$

    \begin{enumerate}[label=(\roman*)]
        \item[(a)] $Q_1$ is equivalent to $Q_2$ and $Q'_1$ is equivalent to $Q'_2.$
        \item[(b)] $\displaystyle Q_1 \cap Q_2=<\bar{h}>, \, Q_1+Q_2=\frac{\Z_{121}[x]}{<x^p-1>}, \text{where} \\ \bar{h}\in \{9h,20h,31h,42h,53h,64h,75h,86h,97h,108h,119h\}.$
        \item[(c)] $Q_1=Q'_1+<\bar{h}>,$ $Q_2=Q'_2+<\bar{h}>.$
        \item[(d)] $|Q_1|=121^{\frac{p+1}{2}}=|Q_2|,$ $|Q'_1|=121^{\frac{p-1}{2}}=|Q'_2|.$
        \item[(e)] $Q_1^\perp = Q'_2,$ $Q_2^\perp = Q'_1.$
        \item[(f)] $Q'_1 \cap Q'_2=\{0\},$ $Q'_1 + Q'_2=< 1-\bar{h}>.$
    \end{enumerate}
\end{theorem}

\begin{theorem}\label{thm:5.15} 
    Let $p$ be a prime of the form $44k+7.$ Then the collection of ideals in $\displaystyle \frac{\Z_{121}[x]}{<x^p-1>}$ corresponding to each case depending on $k$ is given as follows:
    \begin{enumerate}[label=(\roman*)]
    \item If $k=11t,$ let $Q_1=<87+67e_1+106e_2>, \ Q_2=<87+106e_1+67e_2>, \\ 
    Q'_1=<35+15e_1+54e_2>, \ Q'_2=<35+54e_1+15e_2>.$
    
    \item If $k=11t+1,$ let $Q_1=<10+34e_1+106e_2>,\ Q_2=<10+106e_1+34e_2>,\\
    Q'_1=<112+15e_1+87e_2>,\ Q'_2=<112+87e_1+15e_2>.$
    
    \item If $k=11t+2,$ let $Q_1=<54+e_1+106e_2>, \ Q_2=<54+106e_1+e_2>,\\
    Q'_1=<68+15e_1-e_2>, \ Q'_2=<68-e_1+15e_2>.$
    
    \item If $k=11t+3,$ let $Q_1=<98+106e_1+89e_2>, \ Q_2=<98+89e_1+106e_2>,\\
    Q'_1=<24+32e_1+15e_2>, \ Q'_2=<24+15e_1+32e_2>.  $
    
    \item If $k=11t+4,$ let $Q_1=<21+106e_1+56e_2>, \ Q_2=<21+56e_1+106e_2>,\\
    Q'_1=<101+65e_1+15e_2>, \ Q'_2=<101+15e_1+65e_2>.$
    
    \item If $k=11t+5,$ let $Q_1=<65+106e_1+23e_2>, \ Q_2=<65+23e_1+106e_2>,\\
    Q'_1=<57+98e_1+15e_2>, \ Q'_2=<57+15e_1+98e_2>.$
     
    \item If $k=11t+6,$ let $Q_1=<109+106e_1+111e_2>, \ Q_2=<109+111e_1+106e_2>,\\
    Q'_1=<13+10e_1+15e_2>, \ Q'_2=<13+15e_1+10e_2>.$
    
    \item If $k=11t+7,$ let $Q_1=<32+106e_1+78e_2>, \ Q_2=<32+78e_1+106e_2>,\\
    Q'_1=<90+43e_1+15e_2>, \ Q'_2=<90+15e_1+43e_2>.  $
    
    \item If $k=11t+8,$ let $Q_1=<76+106e_1+45e_2>, \ Q_2=<76+45e_1+106e_2>,\\
    Q'_1=<46+76e_1+15e_2>, \ Q'_2=<46+15e_1+76e_2>. $
    
    \item If $k=11t+9,$ let $Q_1=<-1+12e_1+106e_2>, \ Q_2=<-1+106e_1+12e_2>,\\
    Q'_1=<2+15e_1+109e_2>, \ Q'_2=<2+109e_1+15e_2>.$
    
    \item If $k=11t+10,$ let $Q_1=<43+100e_1+106e_2>, \ Q_2=<43+106e_1+100e_2>,\\
    Q'_1=<79+15e_1+21e_2>, \ Q'_2=<79+21e_1+15e_2>.$
    \end{enumerate}
    Then the following holds for $\Z_{121}$- QR codes $Q_1,Q_2,Q'_1$ and $Q'_2.$

    \begin{enumerate}[label=(\roman*)]
        \item[(a)] $Q_1$ is equivalent to $Q_2$ and $Q'_1$ is equivalent to $Q'_2.$
        \item[(b)] $\displaystyle Q_1 \cap Q_2=<\bar{h}>, \, Q_1+Q_2=\frac{\Z_{121}[x]}{<x^p-1>}, \text{where} \\ \bar{h}\in \{8h,19h,30h,41h,52h,63h,74h,85h,96h,107h,118h\}.$
        \item[(c)] $Q_1=Q'_1+<\bar{h}>,$ $Q_2=Q'_2+<\bar{h}>.$
        \item[(d)] $|Q_1|=121^{\frac{p+1}{2}}=|Q_2|,$ $|Q'_1|=121^{\frac{p-1}{2}}=|Q'_2|.$
        \item[(e)] $Q_1^\perp = Q'_1,$ $Q_2^\perp = Q'_2.$ Also, $Q'_1$ and $Q'_2$ are self-orthogonal.
        \item[(f)] $Q'_1 \cap Q'_2=\{0\},$ $Q'_1 + Q'_2=< 1-\bar{h}>.$
    \end{enumerate}
\end{theorem}

\begin{theorem}\label{thm:5.16}
    Let $p$ be a prime of the form $44k+9.$ Then the collection of ideals in $\displaystyle \frac{\Z_{121}[x]}{<x^p-1>}$ corresponding to each case depending on $k$ is given as follows:
    \begin{enumerate}[label=(\roman*)]
    \item If $k=11t,$ let $Q_1=<14+94e_1+54e_2>, \ Q_2=<14+54e_1+94e_2>, \\ 
    Q'_1=<108+67e_1+27e_2>, \ Q'_2=<108+27e_1+67e_2>.$
    
    \item If $k=11t+1,$ let $Q_1=<69+10e_1+6e_2>,\ Q_2=<69+6e_1+10e_2>,\\
    Q'_1=<53+115e_1+111e_2>,\ Q'_2=<53+111e_1+115e_2>.$
    
    \item If $k=11t+2,$ let $Q_1=<3+39e_1+87e_2>, \ Q_2=<3+87e_1+39e_2>,\\
    Q'_1=<-2+34e_1+82e_2>, \ Q'_2=<-2+82e_1+34e_2>.$
   
    \item If $k=11t+3,$ let $Q_1=<58+72e_1+43e_2>, \ Q_2=<58+43e_1+72e_2>,\\
    Q'_1=<64+78e_1+49e_2>, \ Q'_2=<64+49e_1+78e_2>.$
    
    \item If $k=11t+4,$ let $ Q_1=<113-e_1+105e_2>, \ Q_2=<113+105e_1-e_2>,\\
    Q'_1=<9+16e_1+e_2>, \ Q'_2=<9+e_1+16e_2>.$
     
    \item If $k=11t+5,$ let $Q_1=<47+76e_1+17e_2>, \ Q_2=<47+17e_1+76e_2>,\\
    Q'_1=<75+104e_1+45e_2>, \ Q'_2=<75+45e_1+104e_2>.$
    
    \item If $k=11t+6,$ let $Q_1=<102+50e_1+32e_2>, \ Q_2=<102+32e_1+50e_2>,\\
    Q'_1=<20+89e_1+71e_2>, \ Q'_2=<20+71e_1+89e_2>.$
    
    \item If $k=11t+7,$ let $Q_1=<36+83e_1+109e_2>, \ Q_2=<36+109e_1+83e_2>,\\
    Q'_1=<86+12e_1+38e_2>, \ Q'_2=<86+38e_1+12e_2>.$
    
    \item If $k=11t+8,$ let $Q_1=<91+65e_1+116e_2>, \ Q_2=<91+116e_1+65e_2>,\\
    Q'_1=<31+5e_1+56e_2>, \ Q'_2=<31+56e_1+5e_2>.$
    
    \item If $k=11t+9,$ let $Q_1=<25+21e_1+28e_2>, \ Q_2=<25+28e_1+21e_2>,\\
    Q'_1=<97+93e_1+100e_2>, \ Q'_2=<97+100e_1+93e_2>.$
   
    \item If $k=11t+10,$ let $Q_1=<80+61e_1+98e_2>, \ Q_2=<80+98e_1+61e_2>,\\
    Q'_1=<42+23e_1+60e_2>, \ Q'_2=<42+60e_1+23e_2>.$
    \end{enumerate}
    Then the following holds for $\Z_{121}$- QR codes $Q_1,Q_2,Q'_1$ and $Q'_2.$

    \begin{enumerate}[label=(\roman*)]
        \item[(a)] $Q_1$ is equivalent to $Q_2$ and $Q'_1$ is equivalent to $Q'_2.$
        \item[(b)] $\displaystyle Q_1 \cap Q_2=<\bar{h}>, \, Q_1+Q_2=\frac{\Z_{121}[x]}{<x^p-1>}, \text{where} \\ \bar{h}\in \{5h,16h,27h,38h,49h,60h,71h,82h,93h,104h,115h\}.$
        \item[(c)] $Q_1=Q'_1+<\bar{h}>,$ $Q_2=Q'_2+<\bar{h}>.$
        \item[(d)] $|Q_1|=121^{\frac{p+1}{2}}=|Q_2|,$ $|Q'_1|=121^{\frac{p-1}{2}}=|Q'_2|.$
        \item[(e)] $Q_1^\perp = Q'_2,$ $Q_2^\perp = Q'_1.$
        \item[(f)] $Q'_1 \cap Q'_2=\{0\},$ $Q'_1 + Q'_2=< 1-\bar{h}>.$
    \end{enumerate}
\end{theorem}

\begin{theorem}\label{thm:5.17}
    Let $p$ be a prime of the form $44k+19.$ Then the collection of ideals in $\displaystyle \frac{\Z_{121}[x]}{<x^p-1>}$ corresponding to each case depending on $k$ is given as follows:
    \begin{enumerate}[label=(\roman*)]
    \item If $k=11t,$ let $Q_1=<26+107e_1+65e_2>, \ Q_2=<26+65e_1+107e_2>, \\ 
    Q'_1=<96+56e_1+14e_2>, \ Q'_2=<96+14e_1+56e_2>.$
    
    \item If $k=11t+1,$ let $Q_1=<37+43e_1+30e_2>,\ Q_2=<37+30e_1+43e_2>,\\
    Q'_1=<85+91e_1+78e_2>,\ Q'_2=<85+78e_1+91e_2>.$
    
    \item If $k=11t+2,$ let $Q_1=<48+74e_1+21e_2>, \ Q_2=<48+21e_1+74e_2>,\\
    Q'_1=<74+100e_1+47e_2>, \ Q'_2=<74+47e_1+100e_2>.$
    
    \item If $k=11t+3,$ let $Q_1=<59-3e_1-e_2>, \ Q_2=<59-e_1-3e_2>,\\
    Q'_1=<63+e_1+3e_2>, \ Q'_2=<63+3e_1+e_2>.$
   
    \item If $k=11t+4,$ let $Q_1=<70+98e_1+41e_2>, \ Q_2=<70+41e_1+98e_2>,\\
    Q'_1=<52+80e_1+23e_2>, \ Q'_2=<52+23e_1+80e_2>.$
     
    \item If $k=11t+5,$ let $Q_1=<81+85e_1+76e_2>, \ Q_2=<81+76e_1+85e_2>,\\
    Q'_1=<41+45e_1+36e_2>, \ Q'_2=<41+36e_1+45e_2>.$
     
    \item If $k=11t+6,$ let $Q_1=<92+54e_1+8e_2>, \ Q_2=<92+8e_1+54e_2>,\\
    Q'_1=<30+113e_1+67e_2>, \ Q'_2=<30+67e_1+113e_2>.$
    
    \item If $k=11t+7,$ let $Q_1=<103+32e_1+52e_2>, \ Q_2=<103+52e_1+32e_2>,\\
    Q'_1=<19+69e_1+89e_2>, \ Q'_2=<19+89e_1+69e_2>.$
    
    \item If $k=11t+8,$ let $Q_1=<114+96e_1+10e_2>, \ Q_2=<114+10e_1+96e_2>,\\
    Q'_1=<8+111e_1+25e_2>, \ Q'_2=<8+25e_1+111e_2>.$
    
    \item If $k=11t+9,$ let $Q_1=<4+109e_1+19e_2>, \ Q_2=<4+19e_1+109e_2>,\\
    Q'_1=<-3+102e_1+12e_2>, \ Q'_2=<-3+12e_1+102e_2>.$
    
    \item If $k=11t+10,$ let $Q_1=<15+63e_1+87e_2>, \ Q_2=<15+87e_1+63e_2>,\\
    Q'_1=<107+34e_1+58e_2>, \ Q'_2=<107+58e_1+34e_2>.$
    \end{enumerate}
    Then the following holds for $\Z_{121}$- QR codes $Q_1,Q_2,Q'_1$ and $Q'_2.$

    \begin{enumerate}[label=(\roman*)]
        \item[(a)] $Q_1$ is equivalent to $Q_2$ and $Q'_1$ is equivalent to $Q'_2.$
        \item[(b)] $\displaystyle Q_1 \cap Q_2=<\bar{h}>, \, Q_1+Q_2=\frac{\Z_{121}[x]}{<x^p-1>}, \text{where} \\ \bar{h}\in \{7h,18h,29h,40h,51h,62h,73h,84h,95h,106h,117h\}.$
        \item[(c)] $Q_1=Q'_1+<\bar{h}>,$ $Q_2=Q'_2+<\bar{h}>.$
        \item[(d)] $|Q_1|=121^{\frac{p+1}{2}}=|Q_2|,$ $|Q'_1|=121^{\frac{p-1}{2}}=|Q'_2|.$
        \item[(e)] $Q_1^\perp = Q'_1,$ $Q_2^\perp = Q'_2.$ Also, $Q'_1$ and $Q'_2$ are self-orthogonal.
        \item[(f)] $Q'_1 \cap Q'_2=\{0\},$ $Q'_1 + Q'_2=< 1-\bar{h}>.$
    \end{enumerate}
\end{theorem}

\begin{theorem}\label{thm:5.18}
    Let $p$ be a prime of the form $44k-19.$ Then the collection of ideals in $\displaystyle \frac{\Z_{121}[x]}{<x^p-1>}$ corresponding to each case depending on $k$ is given as follows:
    \begin{enumerate}[label=(\roman*)]
    \item If $k=11t,$ let $Q_1=<107+58e_1+34e_2>, \ Q_2=<107+34e_1+58e_2>, \\ 
    Q'_1=<15+87e_1+63e_2>, \ Q'_2=<15+63e_1+87e_2>.$
    
    \item If $k=11t+1,$ let $Q_1=<-3+12e_1+102e_2>,\ Q_2=<-3+102e_1+12e_2>,\\
    Q'_1=<4+19e_1+109e_2>,\ Q'_2=<4+109e_1+19e_2>.$
    
    \item If $k=11t+2,$ let $Q_1=<8+25e_1+111e_2>, \ Q_2=<8+111e_1+25e_2>,\\
    Q'_1=<114+10e_1+96e_2>, \ Q'_2=<114+96e_1+10e_2>.$
    
    \item If $k=11t+3,$ let $Q_1=<19+89e_1+69e_2>, \ Q_2=<19+69e_1+89e_2>,\\
    Q'_1=<103+52e_1+32e_2>, \ Q'_2=<103+32e_1+52e_2>.$
    
    \item If $k=11t+4,$ let $Q_1=<30+67e_1+113e_2>, \ Q_2=<30+113e_1+67e_2>,\\
    Q'_1=<92+8e_1+54e_2>, \ Q'_2=<92+54e_1+8e_2>.$
     
    \item If $k=11t+5,$ let $Q_1=<41+36e_1+45e_2>, \ Q_2=<41+45e_1+36e_2>,\\
    Q'_1=<81+76e_1+85e_2>, \ Q'_2=<81+85e_1+76e_2>.$
     
    \item If $k=11t+6,$ let $Q_1=<52+23e_1+80e_2>, \ Q_2=<52+80e_1+23e_2>,\\
    Q'_1=<70+41e_1+98e_2>, \ Q'_2=<70+98e_1+41e_2>.$
    
    \item If $k=11t+7,$ let $Q_1=<63+3e_1+e_2>, \ Q_2=<63+e_1+3e_2>,\\
    Q'_1=<59-e_1-3e_2>, \ Q'_2=<59-3e_1-e_2>.$
    
    \item If $k=11t+8,$ let $Q_1=<74+47e_1+100e_2>, \ Q_2=<74+100e_1+47e_2>,\\
    Q'_1=<48+21e_1+74e_2>, \ Q'_2=<48+74e_1+21e_2>.$
    
    \item If $k=11t+9,$ let $Q_1=<85+78e_1+91e_2>, \ Q_2=<85+91e_1+78e_2>,\\
    Q'_1=<37+30e_1+43e_2>, \ Q'_2=<37+43e_1+30e_2>.$
    
    \item If $k=11t+10,$ let $Q_1=<96+14e_1+56e_2>, \ Q_2=<96+56e_1+14e_2>,\\
    Q'_1=<26+65e_1+107e_2>, \ Q'_2=<26+107e_1+65e_2>.$
    \end{enumerate}
    Then the following holds for $\Z_{121}$- QR codes $Q_1,Q_2,Q'_1$ and $Q'_2.$

    \begin{enumerate}[label=(\roman*)]
        \item[(a)] $Q_1$ is equivalent to $Q_2$ and $Q'_1$ is equivalent to $Q'_2.$
        \item[(b)] $\displaystyle Q_1 \cap Q_2=<\bar{h}>, \, Q_1+Q_2=\frac{\Z_{121}[x]}{<x^p-1>}, \text{where} \\ \bar{h}\in \{4h,15h,26h,37h,48h,59h,70h,81h,92h,103h,114h\}.$
        \item[(c)] $Q_1=Q'_1+<\bar{h}>,$ $Q_2=Q'_2+<\bar{h}>.$
        \item[(d)] $|Q_1|=121^{\frac{p+1}{2}}=|Q_2|,$ $|Q'_1|=121^{\frac{p-1}{2}}=|Q'_2|.$
        \item[(e)] $Q_1^\perp = Q'_2,$ $Q_2^\perp = Q'_1.$
        \item[(f)] $Q'_1 \cap Q'_2=\{0\},$ $Q'_1 + Q'_2=< 1-\bar{h}>.$
    \end{enumerate}
\end{theorem}

\begin{theorem}\label{thm:5.19}
    Let $p$ be a prime of the form $44k-9.$ Then the collection of ideals in $\displaystyle \frac{\Z_{121}[x]}{<x^p-1>}$ corresponding to each case depending on $k$ is given as follows:
    \begin{enumerate}[label=(\roman*)]
    \item If $k=11t,$ let $Q_1=<42+60e_1+23e_2>, \ Q_2=<42+23e_1+60e_2>, \\ 
    Q'_1=<80+98e_1+61e_2>, \ Q'_2=<80+61e_1+98e_2>.$
    
    \item If $k=11t+1,$ let $Q_1=<97+100e_1+93e_2>,\ Q_2=<97+93e_1+100e_2>,\\
    Q'_1=<25+28e_1+21e_2>,\ Q'_2=<25+21e_1+28e_2>.$
    
    \item If $k=11t+2,$ let $Q_1=<31+56e_1+5e_2>, \ Q_2=<31+5e_1+56e_2>,\\
    Q'_1=<91+116e_1+65e_2>, \ Q'_2=<91+65e_1+116e_2>.$ 
    
    \item If $k=11t+3,$ let $Q_1=<86+38e_1+12e_2>, \ Q_2=<86+12e_1+38e_2>,\\
    Q'_1=<36+109e_1+83e_2>, \ Q'_2=<36+83e_1+109e_2>.$
    
    \item If $k=11t+4,$ let $Q_1=<20+71e_1+89e_2>, \ Q_2=<20+89e_1+71e_2>,\\
    Q'_1=<102+32e_1+50e_2>, \ Q'_2=<102+50e_1+32e_2>.$
     
    \item If $k=11t+5,$ let $Q_1=<75+45e_1+104e_2>, \ Q_2=<75+104e_1+45e_2>,\\
    Q'_1=<47+17e_1+76e_2>, \ Q'_2=<47+76e_1+17e_2>.$
     
    \item If $k=11t+6,$ let $Q_1=<9+e_1+16e_2>, \ Q_2=<9+16e_1+e_2>,\\
    Q'_1=<113+105e_1-e_2>, \ Q'_2=<113-e_1+105e_2>.$ 
    
    \item If $k=11t+7,$ let $Q_1=<64+49e_1+78e_2>, \ Q_2=<64+78e_1+49e_2>,\\
    Q'_1=<58+43e_1+72e_2>, \ Q'_2=<58+72e_1+43e_2>.$
   
    \item If $k=11t+8,$ let $Q_1=<-2+82e_1+34e_2>, \ Q_2=<-2+34e_1+82e_2>,\\
    Q'_1=<3+87e_1+39e_2>, \ Q'_2=<3+39e_1+87e_2>.$
    
    \item If $k=11t+9,$ let $Q_1=<53+111e_1+115e_2>, \ Q_2=<53+115e_1+111e_2>,\\
    Q'_1=<69+6e_1+10e_2>, \ Q'_2=<69+10e_1+6e_2>.$
    
    \item If $k=11t+10,$ let $Q_1=<108+27e_1+67e_2>, \ Q_2=<108+67e_1+27e_2>,\\
    Q'_1=<14+54e_1+94e_2>, \ Q'_2=<14+94e_1+54e_2>.$
    \end{enumerate}
    Then the following holds for $\Z_{121}$- QR codes $Q_1,Q_2,Q'_1$ and $Q'_2.$

    \begin{enumerate}[label=(\roman*)]
        \item[(a)] $Q_1$ is equivalent to $Q_2$ and $Q'_1$ is equivalent to $Q'_2.$
        \item[(b)] $\displaystyle Q_1 \cap Q_2=<\bar{h}>, \, Q_1+Q_2=\frac{\Z_{121}[x]}{<x^p-1>}, \text{where} \\ \bar{h}\in \{6h,17h,28h,39h,50h,61h,72h,83h,94h,105h,116h\}.$
        \item[(c)] $Q_1=Q'_1+<\bar{h}>,$ $Q_2=Q'_2+<\bar{h}>.$
        \item[(d)] $|Q_1|=121^{\frac{p+1}{2}}=|Q_2|,$ $|Q'_1|=121^{\frac{p-1}{2}}=|Q'_2|.$
        \item[(e)] $Q_1^\perp = Q'_1,$ $Q_2^\perp = Q'_2.$ Also, $Q'_1$ and $Q'_2$ are self-orthogonal.
        \item[(f)] $Q'_1 \cap Q'_2=\{0\},$ $Q'_1 + Q'_2=< 1-\bar{h}>.$
    \end{enumerate}
\end{theorem}

\begin{theorem}\label{thm:5.20}
    Let $p$ be a prime of the form $44k-7.$ Then the collection of ideals in $\displaystyle \frac{\Z_{121}[x]}{<x^p-1>}$ corresponding to each case depending on $k$ is given as follows:
    \begin{enumerate}[label=(\roman*)]
    \item If $k=11t,$ let $Q_1=<79+21e_1+15e_2>, \ Q_2=<79+15e_1+21e_2>, \\ 
    Q'_1=<43+106e_1+100e_2>, \ Q'_2=<43+100e_1+106e_2>.$
    
    \item If $k=11t+1,$ let $Q_1=<2+109e_1+15e_2>,\ Q_2=<2+15e_1+109e_2>,\\
    Q'_1=<-1+106e_1+12e_2>,\ Q'_2=<-1+12e_1+106e_2>.$
    
    \item If $k=11t+2,$ let $Q_1=<46+15e_1+76e_2>, \ Q_2=<46+76e_1+15e_2>,\\
    Q'_1=<76+45e_1+106e_2>, \ Q'_2=<76+106e_1+45e_2>.$
    
    \item If $k=11t+3,$ let $Q_1=<90+15e_1+43e_2>, \ Q_2=<90+43e_1+15e_2>,\\
    Q'_1=<32+78e_1+106e_2>, \ Q'_2=<32+106e_1+78e_2>.$
    
    \item If $k=11t+4,$ let $Q_1=<13+15e_1+10e_2>, \ Q_2=<13+10e_1+15e_2>,\\
    Q'_1=<109+111e_1+106e_2>, \ Q'_2=<109+106e_1+111e_2>.$
     
    \item If $k=11t+5,$ let $Q_1=<57+15e_1+98e_2>, \ Q_2=<57+98e_1+15e_2>,\\
    Q'_1=<65+23e_1+106e_2>, \ Q'_2=<65+106e_1+23e_2>.$
    
    \item If $k=11t+6,$ let $Q_1=<101+15e_1+65e_2>, \ Q_2=<101+65e_1+15e_2>,\\
    Q'_1=<21+56e_1+106e_2>, \ Q'_2=<21+106e_1+56e_2>.$
    
    \item If $k=11t+7,$ let $Q_1=<24+15e_1+32e_2>, \ Q_2=<24+32e_1+15e_2>,\\
    Q'_1=<98+89e_1+106e_2>, \ Q'_2=<98+106e_1+89e_2>.$
    
    \item If $k=11t+8,$ let $Q_1=<68-e_1+15e_2>, \ Q_2=<68+15e_1-e_2>,\\
    Q'_1=<54+106e_1+e_2>, \ Q'_2=<54+e_1+106e_2>.$
    
    \item If $k=11t+9,$ let $Q_1=<112+87e_1+15e_2>, \ Q_2=<112+15e_1+87e_2>,\\
    Q'_1=<10+106e_1+34e_2>, \ Q'_2=<10+34e_1+106e_2>.$
    
    \item If $k=11t+10,$ let $Q_1=<35+54e_1+15e_2>, \ Q_2=<35+15e_1+54e_2>,\\
    Q'_1=<87+106e_1+67e_2>, \ Q'_2=<87+67e_1+106e_2>.$
    \end{enumerate}
    Then the following holds for $\Z_{121}$- QR codes $Q_1,Q_2,Q'_1$ and $Q'_2.$

    \begin{enumerate}[label=(\roman*)]
        \item[(a)] $Q_1$ is equivalent to $Q_2$ and $Q'_1$ is equivalent to $Q'_2.$
        \item[(b)] $\displaystyle Q_1 \cap Q_2=<\bar{h}>, \, Q_1+Q_2=\frac{\Z_{121}[x]}{<x^p-1>}, \text{where} \\ \bar{h}\in \{3h,14h,25h,36h,47h,58h,69h,80h,91h,102h,113h\}.$
        \item[(c)] $Q_1=Q'_1+<\bar{h}>,$ $Q_2=Q'_2+<\bar{h}>.$
        \item[(d)] $|Q_1|=121^{\frac{p+1}{2}}=|Q_2|,$ $|Q'_1|=121^{\frac{p-1}{2}}=|Q'_2|.$
        \item[(e)] $Q_1^\perp = Q'_2,$ $Q_2^\perp = Q'_1.$
        \item[(f)] $Q'_1 \cap Q'_2=\{0\},$ $Q'_1 + Q'_2=< 1-\bar{h}>.$
    \end{enumerate}
\end{theorem}

\begin{theorem}\label{thm:5.21}
    Let $p$ be a prime of the form $44k-5.$ Then the collection of ideals in $\displaystyle \frac{\Z_{121}[x]}{<x^p-1>}$ corresponding to each case depending on $k$ is given as follows:
    \begin{enumerate}[label=(\roman*)]
    \item If $k=11t,$ let $Q_1=<106+82e_1+8e_2>, \ Q_2=<106+8e_1+82e_2>, \\ 
    Q'_1=<16+113e_1+39e_2>, \ Q'_2=<16+39e_1+113e_2>.$
    
    \item If $k=11t+1,$ let $Q_1=<18+96e_1+60e_2>,\ Q_2=<18+60e_1+96e_2>,\\
    Q'_1=<104+61e_1+25e_2>,\ Q'_2=<104+25e_1+61e_2>.$
    
    \item If $k=11t+2,$ let $Q_1=<51+38e_1+63e_2>, \ Q_2=<51+63e_1+38e_2>,\\
    Q'_1=<71+58e_1+83e_2>, \ Q'_2=<71+83e_1+58e_2>.$
    
    \item If $k=11t+3,$ let $Q_1=<84+30e_1+16e_2>, \ Q_2=<84+16e_1+30e_2>,\\
    Q'_1=<38+105e_1+91e_2>, \ Q'_2=<38+91e_1+105e_2>.$
    
    \item If $k=11t+4,$ let $Q_1=<-4-6e_1-3e_2>, \ Q_2=<-4-3e_1-6e_2>,\\
    Q'_1=<5+3e_1+6e_2>, \ Q'_2=<5+6e_1+3e_2>.$
    
    \item If $k=11t+5,$ let $Q_1=<29+85e_1+93e_2>, \ Q_2=<29+93e_1+85e_2>,\\
    Q'_1=<93+28e_1+36e_2>, \ Q'_2=<93+36e_1+28e_2>.$
    
    \item If $k=11t+6,$ let $Q_1=<62+71e_1+52e_2>, \ Q_2=<62+52e_1+71e_2>,\\
    Q'_1=<60+69e_1+50e_2>, \ Q'_2=<60+50e_1+69e_2>.$
    
    \item If $k=11t+7,$ let $Q_1=<95+19e_1+49e_2>, \ Q_2=<95+49e_1+19e_2>,\\
    Q'_1=<27+72e_1+102e_2>, \ Q'_2=<27+102e_1+72e_2>.$
    
    \item If $k=11t+8,$ let $Q_1=<7+107e_1+27e_2>, \ Q_2=<7+27e_1+107e_2>,\\
    Q'_1=<115+94e_1+14e_2>, \ Q'_2=<115+14e_1+94e_2>.$
    
    \item If $k=11t+9,$ let $Q_1=<40+5e_1+74e_2>, \ Q_2=<40+74e_1+5e_2>,\\
    Q'_1=<82+47e_1+116e_2>, \ Q'_2=<82+116e_1+47e_2>.$
    
    \item If $k=11t+10,$ let $Q_1=<73+41e_1+104e_2>, \ Q_2=<73+104e_1+41e_2>,\\
    Q'_1=<49+17e_1+80e_2>, \ Q'_2=<49+80e_1+17e_2>.$
    \end{enumerate}
    Then the following holds for $\Z_{121}$- QR codes $Q_1,Q_2,Q'_1$ and $Q'_2.$

    \begin{enumerate}[label=(\roman*)]
        \item[(a)] $Q_1$ is equivalent to $Q_2$ and $Q'_1$ is equivalent to $Q'_2.$
        \item[(b)] $\displaystyle Q_1 \cap Q_2=<\bar{h}>, \, Q_1+Q_2=\frac{\Z_{121}[x]}{<x^p-1>}, \text{where} \\ \bar{h}\in \{2h,13h,24h,35h,46h,57h,68h,79h,90h,101h,112h\}.$
        \item[(c)] $Q_1=Q'_1+<\bar{h}>,$ $Q_2=Q'_2+<\bar{h}>.$
        \item[(d)] $|Q_1|=121^{\frac{p+1}{2}}=|Q_2|,$ $|Q'_1|=121^{\frac{p-1}{2}}=|Q'_2|.$
        \item[(e)] $Q_1^\perp = Q'_1,$ $Q_2^\perp = Q'_2.$ Also, $Q'_1$ and $Q'_2$ are self-orthogonal.
        \item[(f)] $Q'_1 \cap Q'_2=\{0\},$ $Q'_1 + Q'_2=< 1-\bar{h}>.$
    \end{enumerate}
\end{theorem}

\begin{theorem}\label{thm:5.22}
    Let $p$ be a prime of the form $44k-1.$ Then the collection of ideals in $\displaystyle \frac{\Z_{121}[x]}{<x^p-1>}$ corresponding to each case depending on $k$ is given as follows:
    \begin{enumerate}[label=(\roman*)]
    \item If $k=11t,$ let $Q_1=<99+87e_1+110_2>, \ Q_2=<99+110e_1+87e_2>, \\ 
    Q'_1=<23+11e_1+34e_2>, \ Q'_2=<23+34e_1+11e_2>.$
    
    \item If $k=11t+1,$ let $Q_1=<77+54e_1+99e_2>,\ Q_2=<77+99e_1+54e_2>,\\
    Q'_1=<45+22e_1+67e_2>,\ Q'_2=<45+67e_1+22e_2>.$
    
    \item If $k=11t+2,$ let $Q_1=<55+88e_1+21e_2>, \ Q_2=<55+21e_1+88e_2>,\\
    Q'_1=<67+100e_1+33e_2>, \ Q'_2=<67+33e_1+100e_2>.$
    
    \item If $k=11t+3,$ let $Q_1=<33+77e_1+109e_2>, \ Q_2=<33+109e_1+77e_2>,\\
    Q'_1=<89+12e_1+44e_2>, \ Q'_2=<89+44e_1+12e_2>.$
    
    \item If $k=11t+4,$ let $Q_1=<11+66e_1+76e_2>, \ Q_2=<11+76e_1+66e_2>,\\
    Q'_1=<111+45e_1+55e_2>, \ Q'_2=<111+55e_1+45e_2>.$
     
    \item If $k=11t+5,$ let $Q_1=<110+43e_1+55e_2>, \ Q_2=<110+55e_1+43e_2>,\\
    Q'_1=<12+66e_1+78e_2>, \ Q'_2=<12+78e_1+66e_2>.$
    
    \item If $k=11t+6,$ let $Q_1=<88+10e_1+44e_2>, \ Q_2=<88+44e_1+10e_2>,\\
    Q'_1=<34+77e_1+111e_2>, \ Q'_2=<34+111e_1+77e_2>.$
    
    \item If $k=11t+7,$ let $Q_1=<66+98e_1+33e_2>, \ Q_2=<66+33e_1+98e_2>,\\
    Q'_1=<56+88e_1+23e_2>, \ Q'_2=<56+23e_1+88e_2>.$
    
    \item If $k=11t+8,$ let $Q_1=<44+22e_1+65e_2>, \ Q_2=<44+65e_1+22e_2>,\\
    Q'_1=<78+56e_1+99e_2>, \ Q'_2=<78+99e_1+56e_2>.$
    
    \item If $k=11t+9,$ let $Q_1=<22+11e_1+32e_2>, \ Q_2=<22+32e_1+11e_2>,\\
    Q'_1=<100+89e_1+110e_2>, \ Q'_2=<100+110e_1+89e_2>.$
    
    \item If $k=11t+10,$ let $Q_1=<-e_2>, \ Q_2=<-e_1>,\\
    Q'_1=<1+e_1>, \ Q'_2=<1+e_2>.$
    \end{enumerate}
    Then the following holds for $\Z_{121}$- QR codes $Q_1,Q_2,Q'_1$ and $Q'_2.$

    \begin{enumerate}[label=(\roman*)]
        \item[(a)] $Q_1$ is equivalent to $Q_2$ and $Q'_1$ is equivalent to $Q'_2.$
        \item[(b)] $\displaystyle Q_1 \cap Q_2=<\bar{h}>, \, Q_1+Q_2=\frac{\Z_{121}[x]}{<x^p-1>}, \text{where} \\ \bar{h}\in \{10h,21h,32h,43h,54h,65h,76h,87h,98h,109h,120h\}.$
        \item[(c)] $Q_1=Q'_1+<\bar{h}>,$ $Q_2=Q'_2+<\bar{h}>.$
        \item[(d)] $|Q_1|=121^{\frac{p+1}{2}}=|Q_2|,$ $|Q'_1|=121^{\frac{p-1}{2}}=|Q'_2|.$
        \item[(e)] $Q_1^\perp = Q'_1,$ $Q_2^\perp = Q'_2.$ Also, $Q'_1$ and $Q'_2$ are self-orthogonal.
        \item[(f)] $Q'_1 \cap Q'_2=\{0\},$ $Q'_1 + Q'_2=< 1-\bar{h}>.$
    \end{enumerate}
\end{theorem}

\section{\textbf{Extended QR codes over \texorpdfstring{\ensuremath{\Z_{121}}}{Z121}}}
We extend our theory to the extended quadratic residue codes over $\Z_{121}.$ 

\begin{definition}
    The extended code of a quadratic residue code $C$ over $\Z_{121}$ denoted by $\hat{C},$ is the code obtained by adding a specific column to the generator matrix of $C, i.e.,$
    \[
    \hat{C}=\{\hat{c} \, | \, c \in C\},
    \]
    where $\hat{c}=(c_\infty,c_0,c_1,\dots,c_{p-1})$ such that $c_\infty+c_0+c_1+\cdots+c_{p-1} \equiv 0 \pmod {121}.$
\end{definition}

The following theorems give classification of self-dual codes.

\begin{theorem}
   Let $p$ be a prime of the form $44k+1.$ Consider the $\Z_{121}$- quadratic residue codes $Q_1$ and $Q_2$ as described in Theorem-\ref{thm:5.13}, and let $\hat{Q_1}$ and $\hat{Q_2}$ denote their extended codes. Then the ring $\mathbb{Z}_{121}[x]/\langle x^p-1 \rangle$ admits an orthogonal direct sum decomposition into these extended codes, that is,
\[
    \frac{\mathbb{Z}_{121}[x]}{\langle x^p-1 \rangle} = \hat{Q}_1 \oplus \hat{Q}_2,
\]
where the dual of $\hat{Q}_1$ is $\hat{Q}_2$, and the dual of $\hat{Q}_2$ is $\hat{Q}_1,$ i.e., $\hat{Q}_1^{\perp} = \hat{Q}_2$ and $\hat{Q}_2^{\perp} = \hat{Q}_1.$
\end{theorem}

\begin{proof}
    We only prove the case where $k=11t+1,$ the proofs for the other cases are similar.  \\
    By Theorem-\ref{thm:5.13}, we have $Q_1=Q'_1+<\bar{h}>=Q'_1+<78h>$ and $\hat{Q_1}$ has the following generator matrix of order $\frac{p+1}{2} \times (p+1)$:-  
\[
\begin{array}{c}
r_0\\
r_1\\
\\
\vdots\\
r_\infty
\end{array}
\begin{pmatrix}
0 &   &   &   &   &   \\[6pt]
0 &   &   & G_1'  &   &   \\[6pt]
\vdots &   &   &   &   &   \\[6pt]
120 & 78 & 78 & 78 & \cdots & 78
\end{pmatrix}    
,\]
where $r_i$ denotes the $i^{th}$ row of the matrix and each row of $G'_1$ is a cyclic shift of $Q'_1=<22+32e_1+11e_2>.$ \\
Given that $G'_1$ is a generator matrix for $Q'_1$ and $Q_2^\perp =Q'_1.$ Since every row in the matrix above is orthogonal to every row in the matrix defining $\hat{Q}_2,$ $\hat{Q}_2 \subseteq \hat{Q}_1^\perp.$ Given that, $|\hat{Q}_2|=|\hat{Q}_1^\perp|,$ thus $\hat{Q}_2 = \hat{Q}_1^\perp.$ 
\end{proof}

Note that this duality is with respect to the standard inner product on $\mathbb{Z}_{121}^p$ and reflects the self-orthogonality and completeness of the decomposition in the ambient module.

\begin{theorem}
   Let $p$ be a prime of the form $44k-1.$ Consider the $\Z_{121}$- quadratic residue codes $Q_1$ and $Q_2$ as described in Theorem-\ref{thm:5.22}, and $\hat{Q_1},\hat{Q_2}$ denote their extended codes. Then $\hat{Q_1}$ and $\hat{Q_2}$ are self-dual codes, i.e., $\hat{Q}_1^{\perp} = \hat{Q}_1$ and $\hat{Q}_2^{\perp} = \hat{Q}_2.$   
\end{theorem}

\begin{proof}
    We only prove the case where $k=11t,$ the proofs for the other cases are similar.\\
    By Theorem-\ref{thm:5.22}, we have $Q_1=Q'_1+<\bar{h}>=Q'_1+<76h>$ and $\hat{Q_1}$ has the following generator matrix of order $\frac{p+1}{2} \times (p+1)$:-
\[
\begin{array}{c}
r_0\\
r_1\\
\\
\vdots\\
r_\infty
\end{array}
\begin{pmatrix}
0 &   &   &   &   &   \\[6pt]
0 &   &  & G_1'  &   &   \\[6pt]
\vdots &   &   &   &   &   \\[6pt]
120 & 76 & 76 & 76 & \cdots & 76
\end{pmatrix}
,\]

where $r_i$ denotes the $i^{th}$ row of the matrix and each row of $G'_1$ is a cyclic shift of the $Q'_1=<23+11e_1+34e_2>.$ \\
Given that $G'_1$ is a generator matrix for code $Q'_1$ and $Q'_1$ is self-orthogonal, the rows of $G'_1$ are orthogonal to each other and also orthogonal to $76h.$ Since the vector $(120,76h)$ is orthogonal to itself and $|\hat{Q}_1|=|Q_1|=121^{\frac{p+1}{2}},$ $\hat{Q}_1$ is self-orthogonal.\\
Since $ |\hat{Q}^\perp _1|= |\hat{Q}_1|,$ $\hat{Q}_1$ is self-dual. Similarly, we can show that $\hat{Q}_2$ is a self-dual code.
\end{proof}

The proofs of the following theorems are similar to the above two theorems, so we omit the proofs.

\begin{theorem}
   Let $p$ be a prime of the form $44k+5.$ Consider the $\Z_{121}$- quadratic residue codes $Q_1$ and $Q_2$ as described in Theorem- \ref{thm:5.14}, and let $\hat{Q_1}$ and $\hat{Q_2}$ denote their extended codes. Then the ring $\mathbb{Z}_{121}[x]/\langle x^p-1 \rangle$ admits an orthogonal direct sum decomposition into these extended codes, that is,
\[
    \frac{\mathbb{Z}_{121}[x]}{\langle x^p-1 \rangle} = \hat{Q}_1 \oplus \hat{Q}_2,
\]
where the dual of $\hat{Q}_1$ is $\hat{Q}_2$, and the dual of $\hat{Q}_2$ is $\hat{Q}_1.$
\end{theorem}

\begin{theorem}
   Let $p$ be a prime of the form $44k-5.$ Consider the $\Z_{121}$- quadratic residue codes $Q_1$ and $Q_2$ as described in Theorem- \ref{thm:5.21}, and let $\hat{Q_1}$ and $\hat{Q_2}$ denote their extended codes. Then $\hat{Q_1}$ and $\hat{Q_2}$ are self-dual codes.
\end{theorem}

\begin{theorem}
   Let $p$ be a prime of the form $44k+9.$ Consider the $\Z_{121}$- quadratic residue codes $Q_1$ and $Q_2$ as described in Theorem- \ref{thm:5.16}, and let $\hat{Q_1}$ and $\hat{Q_2}$ denote their extended codes. Then the ring $\mathbb{Z}_{121}[x]/\langle x^p-1 \rangle$ admits an orthogonal direct sum decomposition into these extended codes, that is,
\[
    \frac{\mathbb{Z}_{121}[x]}{\langle x^p-1 \rangle} = \hat{Q}_1 \oplus \hat{Q}_2,
\]
where the dual of $\hat{Q}_1$ is $\hat{Q}_2$, and the dual of $\hat{Q}_2$ is $\hat{Q}_1.$
\end{theorem}

\begin{theorem}
   Let $p$ be a prime of the form $44k-9.$ Consider the $\Z_{121}$- quadratic residue codes $Q_1$ and $Q_2$ as described in Theorem- \ref{thm:5.19}, and let $\hat{Q_1}$ and $\hat{Q_2}$ denote their extended codes. Then $\hat{Q_1}$ and $\hat{Q_2}$ are self-dual codes.
\end{theorem}

\begin{theorem}
   Let $p$ be a prime of the form $44k-7.$ Consider the $\Z_{121}$- quadratic residue codes $Q_1$ and $Q_2$ as described in Theorem- \ref{thm:5.20}, and let $\hat{Q_1}$ and $\hat{Q_2}$ denote their extended codes. Then the ring $\mathbb{Z}_{121}[x]/\langle x^p-1 \rangle$ admits an orthogonal direct sum decomposition into these extended codes, that is,
\[
    \frac{\mathbb{Z}_{121}[x]}{\langle x^p-1 \rangle} = \hat{Q}_1 \oplus \hat{Q}_2,
\]
where the dual of $\hat{Q}_1$ is $\hat{Q}_2$, and the dual of $\hat{Q}_2$ is $\hat{Q}_1.$
\end{theorem}

\begin{theorem}
   Let $p$ be a prime of the form $44k+7.$ Consider the $\Z_{121}$- quadratic residue codes $Q_1$ and $Q_2$ as described in Theorem- \ref{thm:5.15}, and let $\hat{Q_1}$ and $\hat{Q_2}$ denote their extended codes. Then $\hat{Q_1}$ and $\hat{Q_2}$ are self-dual codes.
\end{theorem}

\begin{theorem}
   Let $p$ be a prime of the form $44k-19.$ Consider the $\Z_{121}$- quadratic residue codes $Q_1$ and $Q_2$ as described in Theorem- \ref{thm:5.18}, and let $\hat{Q_1}$ and $\hat{Q_2}$ denote their extended codes. Then the ring $\mathbb{Z}_{121}[x]/\langle x^p-1 \rangle$ admits an orthogonal direct sum decomposition into these extended codes, that is,
\[
    \frac{\mathbb{Z}_{121}[x]}{\langle x^p-1 \rangle} = \hat{Q}_1 \oplus \hat{Q}_2,
\]
where the dual of $\hat{Q}_1$ is $\hat{Q}_2$, and the dual of $\hat{Q}_2$ is $\hat{Q}_1.$
\end{theorem}

\begin{theorem}
   Let $p$ be a prime of the form $44k+19.$ Consider the $\Z_{121}$- quadratic residue codes $Q_1$ and $Q_2$ as described in Theorem- \ref{thm:5.17}, and let $\hat{Q_1}$ and $\hat{Q_2}$ denote their extended codes. Then $\hat{Q_1}$ and $\hat{Q_2}$ are self-dual codes. 
\end{theorem}

Since $Q_1$ and $Q_2$ are equivalent, the extended codes $\hat{Q_1}$ and $\hat{Q_2}$ are equivalent for $p\equiv\pm1\pmod{44},$ $p\equiv\pm5\pmod{44},$ $p\equiv\pm7\pmod{44},$ $p\equiv\pm9\pmod{44},$ and $p\equiv\pm19\pmod{44}.$ Additionally, they are equivalent to $Q'_1$ and $Q'_2.$ Consequently, the group of the extended codes will be the group of either one of the extended codes that will be examined in the next section.

\subsection{\bf{Properties of the Extended QR codes over \ensuremath{\texorpdfstring{\mathbb{Z}_{121}}{Z121}}}\label{sec: extended}}
\hspace{5mm} In this section, we discuss the properties of extended $QR$ codes over $\Z_{121}$ along with some results. Let $\mathcal{Q}$ denotes the quadratic residues modulo $p$ and $a$ be a quadratic residue. \\
Permutations such as shifts and multipliers, that is, $\sigma$ and $\mu _a $respectively, on the set $\{\infty,0,1,2,\dots,p-1\}$ are defined by:-
\begin{align*}
\sigma : i \mapsto i+1 \; (mod\; p), \quad \infty \mapsto \infty \ \ \ \textit{and}\\
\mu_a : i \mapsto ia \; (mod\; p), \quad \infty \mapsto \infty . 
\end{align*}
It is clear that $\sigma$ and $\mu _a$ fix the extended $QR$ codes. Therefore, the group of extended $QR$- codes contains the group generated by $\sigma$ and $\mu _a.$ 

Define the inversion permutation $\rho$ as:
\begin{align*}
\rho : i \xrightarrow{} -\frac{1}{i}(mod \ p), \, \text{followed by multiplication by} -\chi(i).
\end{align*}
Equivalently, on monomials this action is given by 
\begin{align*}
   \rho : x^i \xrightarrow{} -\chi(i)x^{-\frac{1}{i}},
\end{align*}
where the Legendre symbol $\chi(i)$ is given by (\ref{def:gaussian}).

\begin{enumerate}
\item[\textbf{Case 1:}] When $p=44k-1,$ the induced action on the symbols $\infty$ and $0$ depends on the residue class of $k \ (mod \ 11)$ and is described as follows:
\begin{itemize}
    \item if $k=11t,$ then $\infty \xrightarrow{} 0$ followed by multiplication by 43 and $0 \xrightarrow{} \infty$ followed by multiplication by 78
    \item if $k=11t+1,$ then $\infty \xrightarrow{} 0$ followed by multiplication by 87 and $0 \xrightarrow{} \infty$ followed by multiplication by 34
    \item if $k=11t+2,$ then $\infty \xrightarrow{} 0$ followed by multiplication by 10 and $0 \xrightarrow{} \infty$ followed by multiplication by 111
    \item if $k=11t+3,$ then $\infty \xrightarrow{} 0$ followed by multiplication by 54 and $0 \xrightarrow{} \infty$ followed by multiplication by 67
    \item if $k=11t+4,$ then $\infty \xrightarrow{} 0$ followed by multiplication by 98 and $0 \xrightarrow{} \infty$ followed by multiplication by 23
    \item if $k=11t+5,$ then $\infty \xrightarrow{} 0$ followed by multiplication by 21 and $0 \xrightarrow{} \infty$ followed by multiplication by 100
    \item if $k=11t+6,$ then $\infty \xrightarrow{} 0$ followed by multiplication by 65 and $0 \xrightarrow{} \infty$ followed by multiplication by 56
    \item if $k=11t+7,$ then $\infty \xrightarrow{} 0$ followed by multiplication by 109 and $0 \xrightarrow{} \infty$ followed by multiplication by 12
    \item if $k=11t+8,$ then $\infty \xrightarrow{} 0$ followed by multiplication by 32 and $0 \xrightarrow{} \infty$ followed by multiplication by 89
    \item if $k=11t+9,$ then $\infty \xrightarrow{} 0$ followed by multiplication by 76 and $0 \xrightarrow{} \infty$ followed by multiplication by 45
    \item if $k=11t+10,$ then $\infty \xrightarrow{} 0$ followed by multiplication by 120 and $0 \xrightarrow{} \infty$ followed by multiplication by 1.
\end{itemize}

\item[\textbf{Case 2:}] When $p=44k+1,$ the corresponding actions are given by:
\begin{itemize}
    \item if $k=11t,$ then $\infty \xrightarrow{} 0$ followed by multiplication by 1 and $0 \xrightarrow{} \infty$ followed by multiplication by 120
    \item if $k=11t+1,$ then $\infty \xrightarrow{} 0$ followed by multiplication by 45 and $0 \xrightarrow{} \infty$ followed by multiplication by 76
    \item if $k=11t+2,$ then $\infty \xrightarrow{} 0$ followed by multiplication by 89 and $0 \xrightarrow{} \infty$ followed by multiplication by 32
    \item if $k=11t+3,$ then $\infty \xrightarrow{} 0$ followed by multiplication by 12 and $0 \xrightarrow{} \infty$ followed by multiplication by 109
    \item if $k=11t+4,$ then $\infty \xrightarrow{} 0$ followed by multiplication by 56 and $0 \xrightarrow{} \infty$ followed by multiplication by 65
    \item if $k=11t+5,$ then $\infty \xrightarrow{} 0$ followed by multiplication by 100 and $0 \xrightarrow{} \infty$ followed by multiplication by 21
    \item if $k=11t+6,$ then $\infty \xrightarrow{} 0$ followed by multiplication by 23 and $0 \xrightarrow{} \infty$ followed by multiplication by 98
    \item if $k=11t+7,$ then $\infty \xrightarrow{} 0$ followed by multiplication by 67 and $0 \xrightarrow{} \infty$ followed by multiplication by 54
    \item if $k=11t+8,$ then $\infty \xrightarrow{} 0$ followed by multiplication by 111 and $0 \xrightarrow{} \infty$ followed by multiplication by 10
    \item if $k=11t+9,$ then $\infty \xrightarrow{} 0$ followed by multiplication by 34 and $0 \xrightarrow{} \infty$ followed by multiplication by 87
    \item if $k=11t+10,$ then $\infty \xrightarrow{} 0$ followed by multiplication by 78 and $0 \xrightarrow{} \infty$ followed by multiplication by 43.
\end{itemize}
\end{enumerate}

Let $G$ be the group generated by the permutations $\sigma,\mu_a$ and $\rho,$ where $a \in \mathcal{Q}.$ The following lemma identifies this group with $PSL_2(p).$

\begin{lemma} Let $p$ be an odd prime and let
$G=\langle \sigma,\rho,\mu_a : a\in \mathcal{Q} \rangle$. Then $\displaystyle \frac{G}{\{\pm I\}}\cong \mathrm{PSL}_2(p)$ for all odd primes $p.$ Further, 
\begin{enumerate}
\item If $p=44k-1,$ then $G\cong \mathrm{SL}_2(p).$
\item If $p=44k+1,$ then $G\cong \mathrm{PSL}_2(p).$
\end{enumerate}
\end{lemma}

\begin{proof}
Let $G= \ <\sigma, \mu_a,\rho \ | \ a \in \mathcal{Q}>$ be the group generated by certain transformations acting on the extended QR codes.

Permutations correspond to matrices in $SL_2(p)$ as follows:
    \begin{itemize}
        \item \textbf{Translation:} $\sigma : i \mapsto i+1 \; (mod\; p) \quad \leftrightarrow \quad
\begin{pmatrix}
1 & 1 \\
0 & 1
\end{pmatrix}.$ 
    
     \item \textbf{Scaling:} $\mu_a : i \mapsto ia \; (mod\; p) \quad \leftrightarrow \quad
\begin{pmatrix}
a & 0 \\
0 & a^{-1}
\end{pmatrix}.$

    \item \textbf{Inversion:} $\rho_a : x \mapsto -\frac{1}{x} \; (mod\; p) \quad \leftrightarrow \quad
\begin{pmatrix}
0 & -1 \\
1 & 0
\end{pmatrix}.$
    \end{itemize}

\begin{enumerate}
\item[\textbf{Case 1:}] When $p=44k-1,$ $-1$ is not a quadratic residue. Thus, $-I$ is a non-trivial central element, that is, it cannot be expressed using generators.
So, $G \subseteq SL_2(p) .$

As we know, $SL_2(p)$ is generated by translations and inversions, and $\sigma, \mu_a$ and, $\rho$ are exactly these elements. Hence, $G \simeq SL_2(p) .$

Since $\displaystyle PSL_2(p) \simeq \frac{SL_2(p)}{<\pm I>},$ we have $\displaystyle \frac{G}{\pm I} \simeq PSL_2(p).$

\item[\textbf{Case 2:}] When $p=44k+1,$ $-1$ is a quadratic residue. So $\exists$ an element $x$ in $\mathbb{F}_p$ such that $x^2=-1.$ Then,
\[\begin{pmatrix}
    -1 & 0 \\
    0 & -1
\end{pmatrix} 
=
\begin{pmatrix}
    x & 0\\
    0 & x^{-1}
\end{pmatrix}
\begin{pmatrix}
    x & 0\\
    0 & x^{-1}
\end{pmatrix}.\]
Thus, $-I$ is generated by the elements of $G$ and becomes trivial in the induced action. Hence, $G \simeq PSL_2(p).$
\end{enumerate}    
\end{proof}

The following theorem shows that the extended quadratic residue code over $\mathbb{Z}_{121}$ possesses a large permutation automorphism group generated by shifts, multipliers, and inversion, making permutation decoding feasible.  

\begin{theorem}
    Let $G$ be a non-abelian group generated by $\sigma, \mu _a$ and $\rho,$ $i.e.,G=<\sigma, \mu _a , \rho>$ and $a \in \mathcal{Q}.$ Then the automorphism group of the extended $QR$- code contains $G.$
\end{theorem}

\begin{proof}
    Suppose that $p=44k-1$ with $k=11t+2.$ \\
 $\text{Since } \ p=4(11k)-1, \text{therefore }  -1 \in \mathcal{N}.$\\
    Moreover, $Q_1 = Q'_1+<\bar{h}>=Q'_1+<109h>,$ so the extended code $\hat{Q_1}$ is generated by $\displaystyle \frac{p+1}{2}$ rows of the $(p+1) \times (p+1)$ matrix:
\[
\begin{array}{c}
r_0\\
r_1\\
\\
\vdots\\
r_\infty
\end{array}
\begin{pmatrix}
0 &   &   &   &   &   \\[6pt]
0 &   & G_1' &   &   &   \\[6pt]
\vdots &   &   &   &   &   \\[6pt]
120 & 109 & 109 & 109 & \cdots & 109
\end{pmatrix},
\]
where each row of $G'_1$ is obtained as a cyclic shift of the vector $67+100e_1+33e_2.$\\
Since $-1$ is a non-residue, $\displaystyle e_1=\sum_{i\in \mathcal{Q}}x^i$ and $\displaystyle e_2=\sum_{i\in \mathcal{N}}x^i,$ therefore \[\displaystyle \rho(e_1(x))=\rho(\sum_{i\in \mathcal{Q}}x^i)=-\sum_{i\in \mathcal{Q}}\chi(i)x^{-\frac{1}{i}}=-\sum_{j=-\frac{1}{i}\in \mathcal{N}}x^j=-e_2(x)\] and \[\displaystyle \rho(e_2(x))=\rho(\sum_{i\in \mathcal{N}}x^i)=-\sum_{i\in \mathcal{N}}\chi(i)x^{-\frac{1}{i}}=\sum_{j=-\frac{1}{i}\in \mathcal{Q}}x^j=e_1(x).\]

Also, by the definition of $\rho, \ x^0 \xrightarrow{} 111x_\infty$ and $x_\infty \xrightarrow{} 10x^0.$ Now, \begin{itemize}
    \item $r_0 = (0,67+100e_1+33e_2)$
    \item $r_\infty = (120,109+109e_1+109e_2)$
    \item $r_s = (0,67x^s+100\sum x^{q+s}+33\sum x^{n+s})$
    \item $r_{-\frac{1}{s}} = (0,67x^{-\frac{1}{s}}+100\sum x^{q-\frac{1}{s}}+33\sum x^{n-\frac{1}{s}}).$
\end{itemize}
Therefore, \[\rho(r_0)=(67\times111,0\times10-100e_2+33e_1)=(56,33e_1+21e_2)\] and 
\[\rho(r_\infty)=(109\times111,120\times10-109e_2+109e_1)=(-1,111+109e_1+12e_2).\] 
We now verify that $\rho(r_s) \in \hat{Q_1}.$ Throughout the proof, let $q \in \mathcal{Q}$ and $n \in \mathcal{N}.$ \\
We consider two cases:
\begin{itemize}
    \item[\textbf{Case 1:}] $s \in \mathcal{Q}.$ In this case,\\
     $\rho(r_s)=(33\times111,-67\chi(s)x^{-\frac{1}{s}}-100\sum \chi(q+s)x^{-\frac{1}{q+s}}-33\sum \chi(n+s)x^{-\frac{1}{n+s}}).$ \\
    Since $-1$ is a quadratic non-residue and $s$ is a quadratic residue, $0 \in s+\mathcal{N}$ and therefore, $\infty$- coordinate of $\rho(r_s)$ is $33\times111=33 \ (mod \ 121).$ 
    
    We claim that \[\rho(r_s)=120r_{-\frac{1}{s}}+67r_0+116r_\infty+28h.\]
    Since $-1\in \mathcal{N}$ and $s \in \mathcal{Q},$ the non-residue part of $\rho(r_s)+r_{-\frac{1}{s}}$ is:
    \[
    -67x^{-\frac{1}{s}}-100\sum_{q+s\in \mathcal{Q}} x^{-\frac{1}{q+s}}-33\sum_{n+s\in \mathcal{Q}} x^{-\frac{1}{n+s}}+67x^{-\frac{1}{s}}+100\sum_{q-\frac{1}{s}\in \mathcal{N}} x^{q-\frac{1}{s}}+33\sum_{n-\frac{1}{s}\in \mathcal{N}} x^{n-\frac{1}{s}}.
    \]
    Using Theorem-24 (p. 519) \cite{MS}, the set $s+\mathcal{Q}$ has $11r-1$ elements in $\mathcal{Q}$ and $11r$ elements in $\mathcal{N}.$\\
    Since $-1 \in \mathcal{N},$ the set $\{-\frac{1}{q+s}:q+s \neq 0\}$ has $11r-1$ elements in $\mathcal{N}$ and $11r$ elements in $\mathcal{Q}.$\\
    Similarly, the set $\{-\frac{1}{n+s}:n+s \neq 0\}$ has $11r-1$ elements in $\mathcal{N}$ and $11r-1$ elements in $\mathcal{Q}.$\\
    The set $-\frac{1}{s}+\mathcal{Q}$ has $11r-1$ elements in $\mathcal{N}$ and $11r-1$ elements in $\mathcal{Q}$ and one element is $0$ and the set $-\frac{1}{s}+\mathcal{N}$ has $11r-1$ elements in $\mathcal{N}$ and $11r$ elements in $\mathcal{Q}.$\\
    Furthermore, for any $-\frac{1}{q+s} \in \mathcal{N}, \ \exists \ q' \in \mathcal{Q}$ such that $-\frac{1}{q+s}=q'-\frac{1}{s}$ and similarly, for any $-\frac{1}{n+s} \in \mathcal{N}, \ \exists \ n' \in \mathcal{N}$ such that $-\frac{1}{n+s}=n'-\frac{1}{s}.$ Hence the non-residue part simplifies to:
    \[
    -100\sum_{q'-\frac{1}{s} \in \mathcal{N}}x^{q'-\frac{1}{s}}-33\sum_{n'-\frac{1}{s} \in \mathcal{N}}x^{n'-\frac{1}{s}}+100\sum_{q-\frac{1}{s} \in \mathcal{N}}x^{q-\frac{1}{s}}+33\sum_{n-\frac{1}{s} \in \mathcal{N}}x^{n-\frac{1}{s}}=0.
    \]

    Next, the residue part of $\rho(r_s)+r_{-\frac{1}{s}}$ is:
    \[
    100\sum_{q+s \in \mathcal{N}} x^{-\frac{1}{q+s}}+33\sum_{n+s \in \mathcal{N}} x^{-\frac{1}{n+s}}+100\sum_{q-\frac{1}{s}\in \mathcal{Q}} x^{q-\frac{1}{s}}+33\sum_{n-\frac{1}{s}\in \mathcal{Q}} x^{n-\frac{1}{s}}.
    \]
    Since for each $-\frac{1}{q+s} \in \mathcal{Q}, \ \exists \ n' \in \mathcal{N}$ such that $-\frac{1}{q+s}=n'-\frac{1}{s}$ and for each $-\frac{1}{n+s} \in \mathcal{Q}, \ \exists \ q' \in \mathcal{Q}$ such that $-\frac{1}{n+s}=q'-\frac{1}{s},$ we obtain $11r+11r-1=22r-1$ terms, the residue part is equal to:
    \[
    100\sum_{n'-\frac{1}{s} \in \mathcal{Q}}x^{n'-\frac{1}{s}}+33\sum_{q'-\frac{1}{s} \in \mathcal{Q}}x^{q'-\frac{1}{s}}+100\sum_{q-\frac{1}{s} \in \mathcal{Q}}x^{q-\frac{1}{s}}+33\sum_{n-\frac{1}{s} \in \mathcal{Q}}x^{n-\frac{1}{s}}=12e_1.
    \]
    Since $0 \in -\frac{1}{s}+ \mathcal{Q},$ we conclude that
    \[
    \rho(r_s)+r_{-\frac{1}{s}}=(33,100+12e_1)=67r_0+116r_\infty+28h.
    \]
    That is, \[\rho(r_s)=120r_{-\frac{1}{s}}+67r_0+116r_\infty+28h.\]
    Hence, $\rho(r_s) \in \hat{Q_1}.$
    
    \item[\textbf{Case 2:}] $s \in \mathcal{N}.$ In this case,\\
     $\rho(r_s)=(21\times111,-67\chi(s)x^{-\frac{1}{s}}-100\sum \chi(q+s)x^{-\frac{1}{q+s}}-33\sum \chi(n+s)x^{-\frac{1}{n+s}}).$ \\
    Since $0 \in s+\mathcal{Q},\ \infty$- coordinate of $\rho(r_s)$ is $21\times111=32 \ (mod \ 121).$
    
    We claim that \[\rho(r_s)=r_{-\frac{1}{s}}+67r_0-4r_\infty+28h,\]
    where $h$ denotes the all-one vector of length $p+1, \ i.e., \ h=1+e_1+e_2.$
    Since $-1 \in \mathcal{N}$ and $s \in \mathcal{N},$ the residue part of $\rho(r_s)-r_{-\frac{1}{s}}$ is:
    \[
    \displaystyle 67x^{-\frac{1}{s}}+100\sum_{q+s \in \mathcal{N}} x^{-\frac{1}{q+s}}+33\sum_{n+s \in \mathcal{N}} x^{-\frac{1}{n+s}}-\Big(67x^{-\frac{1}{s}}+100\sum_{q-\frac{1}{s}\in \mathcal{Q}} x^{q-\frac{1}{s}}+33\sum_{n-\frac{1}{s}\in \mathcal{Q}} x^{n-\frac{1}{s}}\Big).
    \]
    Using Theorem-24 (p. 519) \cite{MS}, the set $s+\mathcal{Q}$ has $11r-1$ elements in $\mathcal{Q}$ and $11r-1$ elements in $\mathcal{N}$ and one element is $0;$ thus the set $\{-\frac{1}{q+s}:q+s \neq 0\}$ has $11r-1$ elements in $\mathcal{N}$ and $11r-1$ elements in $\mathcal{Q}.$\\
    Similarly, the set $\{-\frac{1}{n+s}:n+s \neq 0\}$ has $11r$ elements in $\mathcal{N}$ and $11r-1$ elements in $\mathcal{Q}.$\\
    The set $-\frac{1}{s}+\mathcal{Q}$ has $11r$ elements in $\mathcal{N}$ and $11r-1$ elements in $\mathcal{Q}$ and the set $-\frac{1}{s}+\mathcal{N}$ has $11r-1$ elements in $\mathcal{N}$ and $11r-1$ elements in $\mathcal{Q}$ and one element is $0.$\\
    Using the same correspondence arguments as above, the residue part is equal to:
    \[
    100\sum_{q'-\frac{1}{s} \in \mathcal{Q}}x^{q'-\frac{1}{s}}+33\sum_{n'-\frac{1}{s} \in \mathcal{Q}}x^{n'-\frac{1}{s}}-100\sum_{q-\frac{1}{s} \in \mathcal{Q}}x^{q-\frac{1}{s}}-33\sum_{n-\frac{1}{s} \in \mathcal{Q}}x^{n-\frac{1}{s}}=0.
    \]

    Furthermore, the non-residue part of $\rho(r_s)-r_{-\frac{1}{s}}$ is:
    \[
    -100\sum_{q+s \in \mathcal{Q}} x^{-\frac{1}{q+s}}-33\sum_{n+s \in \mathcal{Q}} x^{-\frac{1}{n+s}}-\Big( 100\sum_{q-\frac{1}{s}\in \mathcal{N}} x^{q-\frac{1}{s}}+33\sum_{n-\frac{1}{s}\in \mathcal{N}} x^{n-\frac{1}{s}}\Big).
    \]
    Since for each $-\frac{1}{q+s} \in \mathcal{N}, \ \exists \ n' \in \mathcal{N}$ such that $-\frac{1}{q+s}=n'-\frac{1}{s}$ and for each $-\frac{1}{n+s} \in \mathcal{Q}, \ \exists \ q' \in \mathcal{Q}$ such that $-\frac{1}{n+s}=q'-\frac{1}{s},$ we obtain $11r+11r-1=22r-1$ terms, the non-residue part simplifies to:
    \[
    -133e_2=109e_2.
    \]
    Since $0 \in -\frac{1}{s}+ \mathcal{N},$ we obtain
    \[
    \rho(r_s)-r_{-\frac{1}{s}}=(32,88+109e_2)=67r_0-4r_\infty+28h.
    \]
    That is, \[\rho(r_s)=r_{-\frac{1}{s}}+67r_0-4r_\infty+28h.\]
    Hence, $\rho(r_s) \in \hat{Q_1}.$
    
\end{itemize}    

Now suppose that $p=44k+1$ with $k=11t+1.$ \\
 $\text{Since } \ p=4(11k)+1, \text{therefore }  -1 \in \mathcal{Q}.$\\
    Moreover, $Q_1 = Q'_1+<\bar{h}>=Q'_1+<78h>,$ so the extended code $\hat{Q_1}$ is generated by $\displaystyle \frac{p+1}{2}$ rows of the $(p+1) \times (p+1)$ matrix:
\[
\begin{array}{c}
r_0\\
r_1\\
\\
\vdots\\
r_\infty
\end{array}
\begin{pmatrix}
0 &   &   &   &   &   \\[6pt]
0 &   & G_1' &   &   &   \\[6pt]
\vdots &   &   &   &   &   \\[6pt]
120 & 78 & 78 & 78 & \cdots & 78
\end{pmatrix},
\]
where each row of $G'_1$ is obtained as a cyclic shift of the vector $22+32e_1+11e_2.$\\
Since $-1 \in \mathcal{Q},$ $\displaystyle e_1=\sum_{i\in \mathcal{Q}}x^i$ and $\displaystyle e_2=\sum_{i\in \mathcal{N}}x^i,$ therefore \[\displaystyle \rho(e_1(x))=\rho(\sum_{i\in \mathcal{Q}}x^i)=-\sum_{i\in \mathcal{Q}}\chi(i)x^{-\frac{1}{i}}=-\sum_{j=-\frac{1}{i}\in \mathcal{Q}}x^j=-e_1(x)\] and \[\displaystyle \rho(e_2(x))=\rho(\sum_{i\in \mathcal{N}}x^i)=-\sum_{i\in \mathcal{N}}\chi(i)x^{-\frac{1}{i}}=\sum_{j=-\frac{1}{i}\in \mathcal{N}}x^j=e_2(x).\]

Also, by the definition of $\rho, \ x^0 \xrightarrow{} 76x_\infty$ and $x_\infty \xrightarrow{} 45x^0.$ Now, \begin{itemize}
    \item $r_0 = (0,22+32e_1+11e_2)$
    \item $r_\infty = (120,78+78e_1+78e_2)$
    \item $r_s = (0,22x^s+32\sum x^{q+s}+11\sum x^{n+s})$
    \item $r_{-\frac{1}{s}} = (0,22x^{-\frac{1}{s}}+32\sum x^{q-\frac{1}{s}}+11\sum x^{n-\frac{1}{s}}).$
\end{itemize}
Therefore, \[\rho(r_0)=(22\times76,0\times45-32e_1+11e_2)=(99,89e_1+11e_2)\] and \\
\[\rho(r_\infty)=(78\times76,120\times45-78e_1+78e_2)=(120,76+43e_1+78e_2).\]
We now prove that $\rho(r_s) \in \hat{Q_1}.$ Throughout the proof, let $q \in \mathcal{Q}$ and $n \in \mathcal{N}.$ \\
We consider two cases:
\begin{itemize}
    \item[\textbf{Case 1:}] $s \in \mathcal{Q}.$ In this case,\\
     $\rho(r_s)=(32\times76,-22\chi(s)x^{-\frac{1}{s}}-32\sum \chi(q+s)x^{-\frac{1}{q+s}}-11\sum \chi(n+s)x^{-\frac{1}{n+s}}).$ \\
    Since $-1$ is a quadratic residue and $s$ is a quadratic residue, $0 \in s+\mathcal{Q},$ and therefore, $\infty$- coordinate of $\rho(r_s)$ is $32\times76=12 \ (mod \ 121).$\\ 
    We claim that
    \[
    \rho(r_s)=120r_{-\frac{1}{s}}+21r_0-r_\infty+11h.
    \]
    Since $-1\in \mathcal{Q}$ and $s \in \mathcal{Q},$ the residue part of $\rho(r_s)+r_{-\frac{1}{s}}$ is:
    \[
    -22x^{-\frac{1}{s}}-32\sum_{q+s\in \mathcal{Q}} x^{-\frac{1}{q+s}}-11\sum_{n+s\in \mathcal{Q}} x^{-\frac{1}{n+s}}+22x^{-\frac{1}{s}}+32\sum_{q-\frac{1}{s}\in \mathcal{Q}} x^{q-\frac{1}{s}}+11\sum_{n-\frac{1}{s}\in \mathcal{Q}} x^{n-\frac{1}{s}}.
    \]
    Using Theorem-24 (p. 519) \cite{MS}, the set $s+\mathcal{Q}$ has $11r-1$ elements in $\mathcal{Q},$ $11r$ elements in $\mathcal{N}$ and one element is $0.$\\
    Since $-1 \in \mathcal{Q},$ the set $\{-\frac{1}{q+s}:q+s \neq 0\}$ has $11r-1$ elements in $\mathcal{Q}$ and $11r$ elements in $\mathcal{N}.$\\
    Similarly, the set $\{-\frac{1}{n+s}:n+s \neq 0\}$ has $11r$ elements in $\mathcal{N}$ and $11r$ elements in $\mathcal{Q}.$\\
    The set $-\frac{1}{s}+\mathcal{Q}$ has $11r-1$ elements in $\mathcal{Q},$ $11r$ elements in $\mathcal{N}$ and one element is $0$ and the set $-\frac{1}{s}+\mathcal{N}$ has $11r$ elements in $\mathcal{N}$ and $11r$ elements in $\mathcal{Q}.$\\
   Furthermore, for any $-\frac{1}{q+s} \in \mathcal{Q}, \ \exists \ q' \in \mathcal{Q}$ such that $-\frac{1}{q+s}=q'-\frac{1}{s}$ and for any $-\frac{1}{n+s} \in \mathcal{Q}, \ \exists \ n' \in \mathcal{N}$ such that $-\frac{1}{n+s}=n'-\frac{1}{s}.$ Hence, the residue part simplifies to:
    \[
    -32\sum_{q'-\frac{1}{s} \in \mathcal{Q}}x^{q'-\frac{1}{s}}-11\sum_{n'-\frac{1}{s} \in \mathcal{Q}}x^{n'-\frac{1}{s}}+32\sum_{q-\frac{1}{s} \in \mathcal{Q}}x^{q-\frac{1}{s}}+11\sum_{n-\frac{1}{s} \in \mathcal{Q}}x^{n-\frac{1}{s}}=0.
    \]

    Next, the non-residue part of $\rho(r_s)+r_{-\frac{1}{s}}$ is:
    \[
    32\sum_{q+s \in \mathcal{N}} x^{-\frac{1}{q+s}}+11\sum_{n+s \in \mathcal{N}} x^{-\frac{1}{n+s}}+32\sum_{q-\frac{1}{s}\in \mathcal{N}} x^{q-\frac{1}{s}}+11\sum_{n-\frac{1}{s}\in \mathcal{N}} x^{n-\frac{1}{s}}.
    \]
    Since for each $-\frac{1}{q+s} \in \mathcal{N}, \ \exists \ n' \in \mathcal{N}$ such that $-\frac{1}{q+s}=n'-\frac{1}{s}$ and for each $-\frac{1}{n+s} \in \mathcal{N}, \ \exists \ q' \in \mathcal{Q}$ such that $-\frac{1}{n+s}=q'-\frac{1}{s},$ we obtain $11r+11r=22r$ terms, the non-residue part is equal to:
    \[
    32\sum_{n'-\frac{1}{s} \in \mathcal{N}}x^{n'-\frac{1}{s}}+11\sum_{q'-\frac{1}{s} \in \mathcal{N}}x^{q'-\frac{1}{s}}+32\sum_{q-\frac{1}{s} \in \mathcal{N}}x^{q-\frac{1}{s}}+11\sum_{n-\frac{1}{s} \in \mathcal{N}}x^{n-\frac{1}{s}}=43e_2.
    \]
    Since $0 \in -\frac{1}{s}+ \mathcal{Q},$ we conclude that
    \[
    \rho(r_s)+r_{-\frac{1}{s}}=(12,32+43e_2)=21r_0-r_\infty+11h.
    \]
    That is, \[\rho(r_s)=120r_{-\frac{1}{s}}+21r_0-r_\infty+11h.\]
    Hence, $\rho(r_s) \in \hat{Q_1}.$
    
    \item[\textbf{Case 2:}] $s \in \mathcal{N}.$ In this case,\\
     $\rho(r_s)=(11\times76,-22\chi(s)x^{-\frac{1}{s}}-32\sum \chi(q+s)x^{-\frac{1}{q+s}}-11\sum \chi(n+s)x^{-\frac{1}{n+s}}).$ \\
    Since $0 \in s+\mathcal{N},\ \infty$- coordinate of $\rho(r_s)$ is $11\times76=110 \ (mod \ 121).$ 

    We claim that \[\rho(r_s)=r_{-\frac{1}{s}}+21r_0+11r_\infty.\]
    Since $-1 \in \mathcal{Q}$ and $s \in \mathcal{N},$ the non-residue part of $\rho(r_s)-r_{-\frac{1}{s}}$ is:
    \[
    \displaystyle 22x^{-\frac{1}{s}}+32\sum_{q+s \in \mathcal{N}} x^{-\frac{1}{q+s}}+11\sum_{n+s \in \mathcal{N}} x^{-\frac{1}{n+s}}-\Big(22x^{-\frac{1}{s}}+32\sum_{q-\frac{1}{s}\in \mathcal{N}} x^{q-\frac{1}{s}}+11\sum_{n-\frac{1}{s}\in \mathcal{N}} x^{n-\frac{1}{s}}\Big).
    \]
    Using Theorem-24 (p. 519) \cite{MS}, the set $\{-\frac{1}{q+s}:q+s \neq 0\}$ has $11r$ elements in $\mathcal{N}$ and $11r$ elements in $\mathcal{Q}.$\\
    Similarly, the set $\{-\frac{1}{n+s}:n+s \neq 0\}$ has $11r$ elements in $\mathcal{Q}$ and $11r-1$ elements in $\mathcal{N}.$\\
    The set $-\frac{1}{s}+\mathcal{Q}$ has $11r$ elements in $\mathcal{N}$ and $11r$ elements in $\mathcal{Q}$ and the set $-\frac{1}{s}+\mathcal{N}$ has $11r-1$ elements in $\mathcal{N}$ and $11r$ elements in $\mathcal{Q}$ and one element is $0.$\\
    Using the same correspondence arguments as above, the non-residue part is equal to:
    \[
    32\sum_{q'-\frac{1}{s} \in \mathcal{N}}x^{q'-\frac{1}{s}}+11\sum_{n'-\frac{1}{s} \in \mathcal{N}}x^{n'-\frac{1}{s}}-32\sum_{q-\frac{1}{s} \in \mathcal{N}}x^{q-\frac{1}{s}}-11\sum_{n-\frac{1}{s} \in \mathcal{N}}x^{n-\frac{1}{s}}=0.
    \]

    Furthermore, the residue part of $\rho(r_s)-r_{-\frac{1}{s}}$ is:
    \[
    -32\sum_{q+s \in \mathcal{Q}} x^{-\frac{1}{q+s}}-11\sum_{n+s \in \mathcal{Q}} x^{-\frac{1}{n+s}}-\Big( 32\sum_{q-\frac{1}{s}\in \mathcal{Q}} x^{q-\frac{1}{s}}+11\sum_{n-\frac{1}{s}\in \mathcal{Q}} x^{n-\frac{1}{s}}\Big).
    \]
    Since for each $-\frac{1}{q+s} \in \mathcal{Q}, \ \exists \ n' \in \mathcal{N}$ such that $-\frac{1}{q+s}=n'-\frac{1}{s}$ and for each $-\frac{1}{n+s} \in \mathcal{Q}, \ \exists \ q' \in \mathcal{Q}$ such that $-\frac{1}{n+s}=q'-\frac{1}{s},$ we obtain $11r+11r=22r$ terms, the residue position simplifies to:
    \[
    -43e_1=78e_1.
    \]
    Since $0 \in -\frac{1}{s}+ \mathcal{N},$ we obtain
    \[
    \rho(r_s)-r_{-\frac{1}{s}}=(110,110+78e_1)=21r_0+11r_\infty.
    \]
    That is, \[\rho(r_s)=r_{-\frac{1}{s}}+21r_0+11r_\infty.\] 
    Hence, $\rho(r_s) \in \hat{Q_1}.$
    
\end{itemize}   

By similar proofs, we can obtain results for other cases as well.
\end{proof}

We define a Gray map of the quadratic residue codes over $\mathbb{Z}_{121}$ in the next section.

\section{\bf{The Gray map and Examples}\label{sec: gray}}
We define the Gray map which maps linear codes of length n over $\mathbb{Z}_{121}$ to linear codes of length 11n over $\mathbb{Z}_{11}.$ Let $c=a+bp$ be an element of $\mathbb{Z}_{p^2},$ where $a,b \in \mathbb{Z}_p.$

The Gray map $\phi : \mathbb{Z}_{p^2} \xrightarrow{} \mathbb{Z}_p^p$ is defined as:
\[
\phi (c)=\phi(a+bp)=a(0,1,2,\dots,p-1)+b(1,1,1,\dots,1),
\]
and the homogeneous weight $w_{hom}$ on $\mathbb{Z}_{p^2},$ is defined as:
\[
w_{hom}(c)=
\begin{cases}
p, & \text{if } c \in p\mathbb{Z}_{p^2} \setminus \{0\}, \ i.e., \ a=0,b\neq0, \\[6pt]
p-1, & \text{if } c \notin p\mathbb{Z}_{p^2}, \ i.e., \ a\neq0, \\[6pt]
0, & \text{if } c=0.
\end{cases}
\]
The Gray map for $p=11$ is defined as: 
\[\phi : \mathbb{Z}_{121} \xrightarrow{} \mathbb{Z}_{11}^{11} \ \ \ \textit{such that}\]
\[
\phi(a+11b)=(b,b+a,b+2a,b+3a,b+4a,b+5a,b+6a,b+7a,b+8a,b+9a,b+10a).
\]
This map takes $\mathbb{Z}_{121}$ to $\mathbb{Z}_{11}^{11}$ as follows:

\[
\hspace{-2.2cm}
\renewcommand{\arraystretch}{1.8}\begin{array}{|c|c|c|c|c|c|c|}
\hline
0 \xrightarrow{} (0,0,0,0,0,0,0,0,0,0,0) &  22 \xrightarrow{} (2,2,2,2,2,2,2,2,2,2,2)  & 44 \xrightarrow{} (4,4,4,4,4,4,4,4,4,4,4)\\ \hline
1 \xrightarrow{} (0,1,2,3,4,5,6,7,8,9,10) &  23 \xrightarrow{} (2,3,4,5,6,7,8,9,10,0,1) & 45 \xrightarrow{} (4,5,6,7,8,9,10,0,1,2,3)  \\ \hline
2 \xrightarrow{} (0,2,4,6,8,10,1,3,5,7,9) &  24 \xrightarrow{} (2,4,6,8,10,1,3,5,7,9,0) & 46 \xrightarrow{} (4,6,8,10,1,3,5,7,9,0,2)\\ \hline
3 \xrightarrow{} (0,3,6,9,1,4,7,10,2,5,8) &  25 \xrightarrow{} (2,5,8,0,3,6,9,1,4,7,10) & 47 \xrightarrow{} (4,7,10,2,5,8,0,3,6,9,1)\\ \hline
4 \xrightarrow{} (0,4,8,1,5,9,2,6,10,3,7) &  26 \xrightarrow{} (2,6,10,3,7,0,4,8,1,5,9) & 48 \xrightarrow{} (4,8,1,5,9,2,6,10,3,7,0)\\ \hline
5 \xrightarrow{} (0,5,10,4,9,3,8,2,7,1,6) &  27 \xrightarrow{} (2,7,1,6,0,5,10,4,9,3,8)  & 49 \xrightarrow{} (4,9,3,8,2,7,1,6,0,5,10) \\ \hline
6 \xrightarrow{} (0,6,1,7,2,8,3,9,4,10,5) &  28 \xrightarrow{} (2,8,3,9,4,10,5,0,6,1,7) & 50 \xrightarrow{} (4,10,5,0,6,1,7,2,8,3,9) \\ \hline
7 \xrightarrow{} (0,7,3,10,6,2,9,5,1,8,4) &  29 \xrightarrow{} (2,9,5,1,8,4,0,7,3,10,6) & 51 \xrightarrow{} (4,0,7,3,10,6,2,9,5,1,8)\\ \hline
8 \xrightarrow{} (0,8,5,2,10,7,4,1,9,6,3) &  30 \xrightarrow{} (2,10,7,4,1,9,6,3,0,8,5) & 52 \xrightarrow{} (4,1,9,6,3,0,8,5,2,10,7)\\ \hline
9 \xrightarrow{} (0,9,7,5,3,1,10,8,6,4,2) & 31 \xrightarrow{} (2,0,9,7,5,3,1,10,8,6,4) & 53 \xrightarrow{} (4,2,0,9,7,5,3,1,10,8,6) \\ \hline
10 \xrightarrow{} (0,10,9,8,7,6,5,4,3,2,1) & 32 \xrightarrow{} (2,1,0,10,9,8,7,6,5,4,3) & 54 \xrightarrow{} (4,3,2,1,0,10,9,8,7,6,5)\\ \hline 
11 \xrightarrow{} (1,1,1,1,1,1,1,1,1,1,1)  & 33 \xrightarrow{} (3,3,3,3,3,3,3,3,3,3,3) & 55 \xrightarrow{} (5,5,5,5,5,5,5,5,5,5,5) \\ \hline
12 \xrightarrow{} (1,2,3,4,5,6,7,8,9,10,0) & 34 \xrightarrow{} (3,4,5,6,7,8,9,10,0,1,2) & 56 \xrightarrow{} (5,6,7,8,9,10,0,1,2,3,4)\\ \hline
13 \xrightarrow{} (1,3,5,7,9,0,7,10,2,5,8) & 35 \xrightarrow{} (3,5,7,9,0,2,4,6,8,10,1) & 57 \xrightarrow{} (5,7,9,0,2,4,6,8,10,1,3)\\ \hline
14 \xrightarrow{} (0,4,8,1,5,9,2,6,10,3,7) & 36 \xrightarrow{} (3,6,9,1,4,7,10,2,5,8,0) & 58 \xrightarrow{} (5,8,0,3,6,9,1,4,7,10,2)\\ \hline
15 \xrightarrow{} (0,5,10,4,9,3,8,2,7,1,6) & 37 \xrightarrow{} (3,7,0,4,8,1,5,9,2,6,10) & 59 \xrightarrow{} (5,9,2,6,10,3,7,0,4,8,1)\\ \hline
16 \xrightarrow{} (0,6,1,7,2,8,3,9,4,10,5) & 38 \xrightarrow{} (3,8,2,7,1,6,0,5,10,4,9) & 60 \xrightarrow{} (5,10,4,9,3,8,2,7,1,6,0)\\ \hline
17 \xrightarrow{} (0,7,3,10,6,2,9,5,1,8,4) & 39 \xrightarrow{} (3,9,4,10,5,0,6,1,7,2,8) & 61 \xrightarrow{} (5,0,6,1,7,2,8,3,9,4,10)\\ \hline
18 \xrightarrow{} (0,8,5,2,10,7,4,1,9,6,3) & 40 \xrightarrow{} (3,10,6,2,9,5,1,8,4,0,7) & 62 \xrightarrow{} (5,1,8,4,0,7,3,10,6,2,9)\\ \hline
19 \xrightarrow{} (0,9,7,5,3,1,10,8,6,4,2) & 41 \xrightarrow{} (3,0,8,5,2,10,7,4,1,9,6) & 63 \xrightarrow{} (5,2,10,7,4,1,9,6,3,0,8)\\ \hline
20 \xrightarrow{} (0,10,9,8,7,6,5,4,3,2,1) & 42 \xrightarrow{} (3,1,10,8,6,4,2,0,9,7,5) & 64 \xrightarrow{} (5,3,1,10,8,6,4,2,0,9,7) \\ \hline
21 \xrightarrow{} (0,10,9,8,7,6,5,4,3,2,1) & 43 \xrightarrow{} (3,2,1,0,10,9,8,7,6,5,4) & 65 \xrightarrow{} (5,4,3,2,1,0,10,9,8,7,6) \\ \hline
\end{array}
\]

\[
\hspace{-2.2cm}
\renewcommand{\arraystretch}{1.5}\begin{array}{|c|c|c|c|c|c|c|}
\hline
66 \xrightarrow{} (6,6,6,6,6,6,6,6,6,6,6) & 94 \xrightarrow{} (8,3,9,4,10,5,0,6,1,7,2)   \\ \hline
67 \xrightarrow{} (6,7,8,9,10,0,1,2,3,4,5) & 95 \xrightarrow{} (8,4,0,7,3,10,6,9,2,5,1)   \\ \hline
68 \xrightarrow{} (6,8,10,1,3,5,7,9,0,2,4) & 96 \xrightarrow{} (8,5,2,10,7,4,1,9,6,3,0)   \\ \hline
69 \xrightarrow{} (6,9,1,4,7,10,2,5,8,0,3) & 97 \xrightarrow{} (8,6,4,2,0,9,7,5,3,1,10)   \\ \hline
70 \xrightarrow{} (6,10,3,7,0,4,8,1,5,9,2) & 98 \xrightarrow{} (8,7,6,5,4,3,2,1,0,10,9)   \\ \hline
71 \xrightarrow{} (6,0,5,10,4,9,3,8,2,7,1) & 99 \xrightarrow{} (9,9,9,9,9,9,9,9,9,9,9)   \\ \hline
72 \xrightarrow{} (6,1,7,2,8,3,9,4,10,5,0) & 100 \xrightarrow{} (9,10,0,1,2,3,4,5,6,7,8)    \\ \hline
73 \xrightarrow{} (6,2,9,5,1,8,4,0,7,3,10) & 101 \xrightarrow{} (9,0,2,4,6,8,10,1,3,5,7)    \\ \hline
74 \xrightarrow{} (6,3,0,8,5,2,10,7,4,1,9) & 102 \xrightarrow{} (9,1,4,7,10,2,5,8,0,3,6)    \\ \hline
75 \xrightarrow{} (6,4,2,0,9,7,5,3,1,10,8) & 103 \xrightarrow{} (9,2,6,10,3,7,0,4,8,1,5)    \\ \hline
76 \xrightarrow{} (6,5,4,3,2,1,0,10,9,8,7) & 104 \xrightarrow{} (9,3,8,2,7,1,6,0,5,10,4)    \\ \hline
77 \xrightarrow{} (7,7,7,7,7,7,7,7,7,7,7) & 105 \xrightarrow{} (9,4,10,5,0,6,1,7,2,8,3)   \\ \hline
78 \xrightarrow{} (7,8,9,10,0,1,2,3,4,5,6) & 106 \xrightarrow{} (9,5,1,8,4,0,7,3,10,6,2)    \\ \hline
79 \xrightarrow{} (7,9,0,2,4,6,8,10,1,3,5) & 107 \xrightarrow{} (9,6,3,0,8,5,2,10,7,4,1)    \\ \hline
80 \xrightarrow{} (7,10,2,5,8,0,3,6,9,1,4) & 108 \xrightarrow{} (9,7,5,3,1,10,8,6,4,2,0)    \\ \hline
81 \xrightarrow{} (7,0,4,8,1,5,9,2,6,10,3) & 109 \xrightarrow{} (9,8,7,6,5,4,3,2,1,0,10)    \\ \hline
82 \xrightarrow{} (7,1,6,0,5,10,4,9,3,8,2) & 110 \xrightarrow{} (10,10,10,10,10,10,10,10,10,10,10)    \\ \hline
83 \xrightarrow{} (7,2,8,3,9,4,10,5,0,6,1) & 111 \xrightarrow{} (10,0,1,2,3,4,5,6,7,8,9)    \\ \hline
84 \xrightarrow{} (7,3,10,6,2,9,5,1,8,4,0) & 112 \xrightarrow{} (10,1,3,5,7,9,0,2,4,6,8)    \\ \hline
85 \xrightarrow{} (7,4,1,9,6,3,0,8,5,2,10) & 113 \xrightarrow{} (10,2,5,8,0,3,6,9,1,4,7)    \\ \hline
86 \xrightarrow{} (7,5,3,1,10,8,6,4,2,0,9) & 114 \xrightarrow{} (10,3,7,0,4,8,1,5,9,2,6)    \\ \hline
87 \xrightarrow{} (7,6,5,4,3,2,1,0,10,9,8) & 115 \xrightarrow{} (10,4,9,3,8,2,7,1,6,0,5)    \\ \hline
88 \xrightarrow{} (8,8,8,8,8,8,8,8,8,8,8) &  116 \xrightarrow{} (10,5,0,6,1,7,2,8,3,9,4) \\ \hline
89 \xrightarrow{} (8,9,10,0,1,2,3,4,5,6,7) & 117 \xrightarrow{} (10,6,2,9,5,1,8,4,0,7,3)  \\ \hline
90 \xrightarrow{} (8,10,1,3,5,7,9,0,2,4,6) & 118 \xrightarrow{} (10,7,4,1,9,6,3,0,8,5,2)  \\ \hline
91 \xrightarrow{} (8,0,3,6,9,1,4,7,10,2,5) & 119 \xrightarrow{} (10,8,6,4,2,0,9,7,5,3,1)  \\ \hline
92 \xrightarrow{} (8,1,5,9,2,6,10,3,7,0,4) & 120 \xrightarrow{} (10,9,8,7,6,5,4,3,2,1,0)  \\ \hline
93 \xrightarrow{} (8,2,7,1,6,0,5,10,4,9,3) &   \\ \hline
\end{array}
\]

We have a proposition and a corollary from \cite{LB}.
\begin{proposition}
    The Gray map $\phi : \mathbb{Z}_{p^2}^n \xrightarrow{} \mathbb{Z}_p^{pn}$ defined above is an isometry from $(\mathbb{Z}_{p^2}^n,d_{hom})$ to $(\mathbb{Z}_{p}^{pn},d_H),$ where $d_{hom}$ is the Homogeneous distance on $\mathbb{Z}_{p^2}^{n}$ and $d_H$ is the Hamming distance on $\mathbb{Z}_{p}^{pn}.$   
\end{proposition}

\begin{corollary}
    A code $C$ of length $n$ over $\mathbb{Z}_{p^2}$ is $(1-p)$-cyclic if and only if its Gray image $\phi (C)$ is a cyclic code over $\mathbb{Z}_p$ of length $pn.$ In particular, if $C$ is a linear $(1-p)$- cyclic code of length $n$ over $\mathbb{Z}_{p^2},$ then $\phi(C)$ is a distance-invariant cyclic code over $\mathbb{Z}_p$ of length $pn.$
\end{corollary}

As a direct consequence of the above corollary, when $C$ is a code over $\mathbb{Z}_{121}$ with length $n,$ its Gray image $\phi(C)$ forms a cyclic code of length $11n$ over $\mathbb{Z}_{11}.$ The following proposition talks about self-orthogonality of the Gray image of the code $C.$

\begin{proposition}
    If $C$ is a self-orthogonal cyclic code over $\mathbb{Z}_{121},$ then its Gray image $\phi(C)$ is a self-orthogonal cyclic code over $\mathbb{Z}_{11}.$ That is, $\phi(C) \subseteq \phi(C)^\perp.$ 
\end{proposition}
\begin{proof}
    Assume that $C$ is a self-orthogonal cyclic code over $\mathbb{Z}_{121}$ of length $n.$ That is, $C \subseteq C^\perp.$ Consider two codewords 
    
    $c_1=(a_1+11b_1,a_2+11b_2,\dots,a_n+11b_n)$ and $c_2=(a'_1+11b'_1,a'_2+11b'_2,\dots,a'_n+11b'_n)$ belonging to $C.$ \\
    Since $C$ is self-orthogonal, we have $<c_1,c_2>=0,$ which implies
    \[\displaystyle \sum_{i=1}^n a_ia'_i = \sum_{i=1}^n \left( a_ib'_i+b_ia'_i \right).\]
    Consequently, $a_ia'_i=a_ib'_i+b_ia'_i=0,$ for $i=1,2,\dots,n.$\\
    To establish the result, it suffices to verify that
    \[\phi(c_1)\cdot \phi(c_2)=0.\]
    Indeed, $\displaystyle \phi(c_1)\cdot \phi(c_2)=\sum_{i=1}^n \phi(a_i+11b_i)\cdot \phi(a'_i+11b'_i),$
    so it is enough to show that each summand vanishes.

    By the definition of the Gray map,
    \[\phi(a_i+11b_i)=(b_i,b_i+a_i,\dots,b_i+10a_i)\]
    and
    \[\phi(a'_i+11b'_i)=(b'_i,b'_i+a'_i,\dots,b'_i+10a'_i)\]
    Taking their inner product yields
    \begin{align*}
        \phi(a_i+11b_i)\cdot \phi(a'_i+11b'_i) &=(b_i,b_i+a_i,\dots,b_i+10a_i) \cdot (b'_i,b'_i+a'_i,\dots,b'_i+10a'_i)\\ 
        &=b_ib'_i+(b_i+a_i)(b'_i+a'_i)+\cdots+(b_i+10a_i)(b'_i+10a'_i)\\
        &=11b_ib'_i =0.
    \end{align*}
    Therefore, $\phi(C)$ is self-orthogonal and hence,
    \[\phi(C^\perp)=\phi(C) \subseteq \phi(C)^\perp.\] 
    This completes the proof.
\end{proof}

\begin{remark}
    Note that the Gray image $\phi(C)$ of a self-dual code $C$ over $\Z_{121}$ need not be self-dual code over $\Z_{11}.$
\end{remark}

\begin{example}
    Let $C=\{0,11,22,33,44,55,66,77,88,99,110\}$ be a cyclic code of length $1$ and size $11$ over $\Z_{121}.$
    Clearly, $C$ is a self-dual code. Its Gray image is given by 
    \begin{align*}
    \phi(C)=\phi(C^\perp)=\{(0,0,0,0,0,0,0,0,0,0,0),(1,1,1,1,1,1,1,1,1,1,1), \\(2,2,2,2,2,2,2,2,2,2,2), (3,3,3,3,3,3,3,3,3,3,3), \\ (4,4,4,4,4,4,4,4,4,4,4),(5,5,5,5,5,5,5,5,5,5,5), \\ (6,6,6,6,6,6,6,6,6,6,6),(7,7,7,7,7,7,7,7,7,7,7), \\(8,8,8,8,8,8,8,8,8,8,8),  (9,9,9,9,9,9,9,9,9,9,9), \\ (10,10,10,10,10,10,10,10,10,10,10)\}.
    \end{align*}
    Now consider the element $(1,0,3,0,5,0,0,0,0,2,0) \in \mathbb{Z}_{11}^{11}.$ Since
    \begin{align*}
<(1,0,3,0,5,0,0,0,0,2,0),(1,1,1,1,1,1,1,1,1,1,1)> &= \\<(1,0,3,0,5,0,0,0,0,2,0),(2,2,2,2,2,2,2,2,2,2,2)> &= \\<(1,0,3,0,5,0,0,0,0,2,0),(3,3,3,3,3,3,3,3,3,3,3)>&= \\ <(1,0,3,0,5,0,0,0,0,2,0),(4,4,4,4,4,4,4,4,4,4,4)>
&= \\<(1,0,3,0,5,0,0,0,0,2,0),(5,5,5,5,5,5,5,5,5,5,5)>&= \\ <(1,0,3,0,5,0,0,0,0,2,0),(6,6,6,6,6,6,6,6,6,6,6)>&=\\ <(1,0,3,0,5,0,0,0,0,2,0),(7,7,7,7,7,7,7,7,7,7,7)>&=\\ <(1,0,3,0,5,0,0,0,0,2,0),(8,8,8,8,8,8,8,8,8,8,8)>&= \\ <(1,0,3,0,5,0,0,0,0,2,0),(9,9,9,9,9,9,9,9,9,9,9)>&= \\ <(1,0,3,0,5,0,0,0,0,2,0),(10,10,10,10,10,10,10,10,10,10,10)>&=0,
\end{align*}
    we have, $(1,0,3,0,5,0,0,0,0,2,0) \in \phi(C)^\perp.$
    
    This shows that, $\phi(C) \subsetneq \phi(C)^\perp,$ and hence $\phi(C)$ is not self-dual.
\end{example}

Presented below are some examples of codes arising as Gray images of quadratic residue codes over the ring $\mathbb{Z}_{121},$ under a Gray map that transforms Lee weights into Hamming weights. We utilize Python to calculate the minimum hamming weight \cite{PC}.

\begin{example}
    Let $Q'_1$ be the QR code over $\mathbb{Z}_{121}$ of length $5.$ Define 
    \[Q_{5}=\{1, 4\}\]
    as the set of quadratic residue modulo $5$ and 
    \[N_{5}=\{2, 3\}\]
    as the set of quadratic non-residue modulo $5.$ 
    By Theorem \ref{thm:5.14}, the code $Q'_1$ is a self-orthogonal and is generated by the idempotent 
    \[40+74e_1+5e_2,\]
    where \[\displaystyle e_1(x)=\sum_{i \in Q_5}x^i \ \ \textit{and} \ \ \displaystyle e_2(x)=\sum_{i \in N_5}x^i.\] Then the Gray image of code $Q'_1$ is a self-orthogonal cyclic code over $\mathbb{F}_{11}$ with length $55,$ dimension $5$ and minimum hamming weight $33.$ There are numerous codewords that provide this minimum Hamming weight; among the finest ones is 
    \[
    \begin{aligned}
    (&0, 0, 0, 0, 0, 0, 0, 0, 0, 0, 0, 7, 7, 7, 7, 7, 7, 7, 7, 7, 7, 7, 7, 7, 7, 7, 7, 7, 7, 7, 7, 7, 7, \\
    &0, 0, 0, 0, 0, 0, 0, 0, 0, 0, 0, 1, 1, 1, 1, 1, 1, 1, 1, 1, 1, 1).
    \end{aligned}
    \]
    This result can be verified using the Python program provided in the associated GitHub file.

    In a similar way, one can show that the Gray image of code $Q'_2$ is a $[55,5,33]$ over $\mathbb{F}_{11}.$ \\
    Now, let 
    \[Q_1=<82+116e_1+47e_2>\]
    be a cyclic code of length $5$ over $\mathbb{Z}_{121}.$ Its Gray image is a cyclic code over $\mathbb{F}_{11}$ with length $55,$ dimension $5$ and minimum Hamming weight $30.$ One of the best codewords with this weight is 
    \[
    \begin{aligned}
    (&5, 0, 6, 1, 7, 2, 8, 3, 9, 4, 10, 0, 4, 8, 1, 5, 9, 2, 6, 10, 3, 7, 5, 0, 6, 1, 7, 2, 8, 3, 9, 4, 10,\\
    &0, 0, 0, 0, 0, 0, 0, 0, 0, 0, 0, 0, 0, 0, 0, 0, 0, 0, 0, 0, 0, 0).
    \end{aligned}
    \]

    In a similar way, one can show that the Gray image of code $Q_2$ is a $[55,5,30]$ over $\mathbb{F}_{11}.$
    
\end{example}

\begin{example}
    Let $Q'_1$ be the QR code over $\mathbb{Z}_{121}$ of length $7.$ Define 
    \[Q_{7}=\{1,2,4\}\]
    as the set of quadratic residues modulo $7$ and 
    \[N_{7}=\{3,5,6\}\]
    as the set of quadratic non-residues modulo $7.$ 
    By Theorem \ref{thm:5.15}, the code $Q'_1$ is a self-orthogonal and is generated by the idempotent 
    \[35+15e_1+54e_2,\]
    where \[\displaystyle e_1(x)=\sum_{i \in Q_{7}}x^i \ \ \textit{and} \ \ \displaystyle e_2(x)=\sum_{i \in N_{7}}x^i.\] Then the Gray image of code $Q'_1$ is a self-orthogonal cyclic code over $\mathbb{F}_{11}$ with length $77,$ dimension $7$ and minimum hamming weight $44.$ There are numerous codewords that provide this minimum Hamming weight; among the finest ones is 
    \[
    \begin{aligned}
    (&5, 5, 5, 5, 5, 5, 5, 5, 5, 5, 5, 0, 0, 0, 0, 0, 0, 0, 0, 0, 0, 0, 6, 6, 6, 6, 6, 6, 6, 6, 6, 6, 6,\\
    &2, 2, 2, 2, 2, 2, 2, 2, 2, 2, 2, 0, 0, 0, 0, 0, 0, 0, 0, 0, 0, 0, 0, 0, 0, 0, 0, 0, 0, 0, 0, 0, 0, \\
    &8, 8, 8, 8, 8, 8, 8, 8, 8, 8, 8).
    \end{aligned}
    \]
    This result can be verified using the Python program provided in the associated GitHub file.
    
    In a similar way, one can show that the Gray image of code $Q'_2$ is a $[77,7,44]$ over $\mathbb{F}_{11}.$ \\
    Now, let 
    \[Q_1=<87+67e_1+106e_2>\]
    be a cyclic code of length $7$ over $\mathbb{Z}_{121}.$ Its Gray image is a cyclic code over $\mathbb{F}_{11}$ with length $77,$ dimension $7$ and minimum Hamming weight $40.$ One of the best codewords with this weight is 
    \[
    \begin{aligned}
    (&3, 2, 1, 0, 10, 9, 8, 7, 6, 5, 4, 10, 3, 7, 0, 4, 8, 1, 5, 9, 2, 6, 7, 1, 6, 0, 5, 10, 4, 9, 3, 8, 2,\\
    &8, 9, 10, 0, 1, 2, 3, 4, 5, 6, 7, 0, 0, 0, 0, 0, 0, 0, 0, 0, 0, 0, 0, 0, 0, 0, 0, 0, 0, 0, 0, 0, 0, \\
    &0, 0, 0, 0, 0, 0, 0, 0, 0, 0, 0).
    \end{aligned}
    \]

    In a similar way, one can show that the Gray image of code $Q_2$ is a $[77,7,40]$ over $\mathbb{F}_{11}.$ 
\end{example}

\section{\bf Conclusion\label{sec: conclusion}}
\hspace{5mm}In this paper, we obtained the idempotent generators of quadratic residue codes over the ring $\Z_{121}.$ It was shown that QR codes of prime length $p$ over this ring exist only if $p \equiv \pm 1 \pmod {44}, p \equiv \pm 5 \pmod {44}, p \equiv \pm 7 \pmod {44}, p \equiv \pm 9 \pmod {44},$ and $p \equiv \pm 19 \pmod {44}.$ These codes were defined in terms of their idempotent generators. Using this method, one can easily find the idempotent generators of QR codes over $\Z_q^n,$ for $q>2$ to be a prime and $n$ to be a positive integer. Later, we discussed their extension codes. We described the Gray map of quadratic residue codes over the ring $\Z_{121}$ and obtained some new codes that are the Gray images of quadratic residue codes. Examples of such codes include $[55,5,33]$ and $[77,7,44].$ Our approach highlights the interplay between ring and coding theory, where the idempotent elements serve as a powerful tool in the design and analysis of quadratic residue codes. Future research may extend these methods to more complex codes or explore the applications of different algebraic structures in the construction of codes with specific desired properties.

\section{Data Availability} 	
The authors confirm that their manuscript has no associated data.

\section{Competing Interests}
The authors confirm that they have no competing interest. 

\bibliographystyle{plain}

\end{document}